\documentclass[12pt]{amsart}

\usepackage{etoolbox}
\usepackage[stable]{footmisc}
\usepackage[latin1]{inputenc}
\usepackage{graphicx}
\usepackage{epstopdf}
\usepackage{enumerate}
\usepackage{cite}
\usepackage{amssymb,amsmath}
\usepackage{amsthm}
\usepackage{multirow}
\usepackage{hyperref}
\usepackage[margin=3 cm]{geometry}
\usepackage{verbatim}
\usepackage[noend]{algpseudocode}
\usepackage{chngcntr}
\usepackage{booktabs}
\usepackage{multirow}
\usepackage{algorithm}
\usepackage{sidecap}
\usepackage{blkarray}

\usepackage{tikz}
\usetikzlibrary{arrows,backgrounds}
\usepackage[all]{xy}

\newcommand{\R}{{\mathbb R}}

\newcommand{\Z}{{\mathbb Z}}

\newcommand{\perm}{\mathrm{Perm}}
\newcommand{\myperm}{\mathit{perm}}
\renewcommand{\det}{\mathrm{Det}}
\newcommand{\mydet}{\mathit{det}}
\newcommand{\sgn}{\mathrm{sgn}}
\newcommand{\mdisc}{\mathrm{Disc}}
\newcommand{\mydisc}{\mathit{disc}}
\newcommand{\mvol}{{\mathrm{MVol}}}
\newcommand{\vol}{\mathit{vol}}
\newcommand{\tw}{\mathrm{tw}}
\newcommand{\sharpp}{\#\rm{P}}
\newcommand{\card}[1]{\ensuremath{\left\vert#1\right\vert}}

\renewcommand{\subset}{\subseteq}
\renewcommand{\supset}{\supseteq}

\newtheorem{thm}{Theorem}
\newtheorem{prop}[thm]{Proposition}
\newtheorem{cor}[thm]{Corollary}
\newtheorem{lem}[thm]{Lemma}

\newtheorem*{ques}{Question}
\newtheorem*{notation}{Notation}
\theoremstyle{definition}
\newtheorem{defn}{Definition}[section]
\newtheorem{exmp}{Example}[section]
\theoremstyle{remark}
\newtheorem{rem}{Remark}[section]

\let\oldtabular\tabular
\renewcommand{\tabular}{\footnotesize\oldtabular}

\graphicspath{ {images/} }

\title[]{An efficient tree decomposition method for permanents and mixed discriminants}
\date{\today}
\author{Diego Cifuentes} 
\address{
Laboratory for Information and Decision Systems (LIDS), 
Massachusetts Institute of Technology, Cambridge MA 02139, USA}
\email{diegcif@mit.edu}

\author{Pablo A. Parrilo}
\address{
Laboratory for Information and Decision Systems (LIDS), 
Massachusetts Institute of Technology, Cambridge MA 02139, USA}
\email{parrilo@mit.edu}

\keywords {Permanent, Structured array, Mixed discriminant, Treewidth}

\begin{document}
\maketitle

\begin{abstract}
  We present an efficient algorithm to compute permanents, mixed discriminants and hyperdeterminants of structured matrices and multidimensional arrays (tensors).
  We describe the sparsity structure of an array in terms of a graph,
  and we assume that its treewidth, denoted as $\omega$, is small.
  Our algorithm requires $\widetilde{O}(n\, 2^{\omega})$ arithmetic operations to compute permanents, and $\widetilde{O}(n^2 + n\, 3^{\omega})$ for mixed discriminants and hyperdeterminants.
  We finally show that mixed volume computation continues to be hard under bounded treewidth assumptions.
\end{abstract}

% Corrections Pablo -- June 1
% Subset convolution more explicit now
% Complexity determinant 3^n

\section{Introduction}

The \emph{permanent} of a $n\times n$ matrix $M$ is defined as
\begin{align*}
  \perm(M) := \sum_{\pi} \prod_{i=1}^n M_{i,\pi(i)}
\end{align*}
where the sum is over all permutations $\pi$ of the numbers $1,\ldots,n$.
Computing the permanent is $\sharpp$-hard \cite{Valiant1979}, which means that it is unlikely that it can be done efficiently for arbitrary matrices.
As a consequence, research on this problem tends to fall into two categories: algorithms to approximate the permanent, and exact algorithms that assume some structure of the matrix.
This paper lies in the second category.
We further study related problems in structured higher dimensional arrays, such as mixed discriminants, hyperdeterminants and mixed volumes.

%Given a sparse matrix $M$ there is a natural bipartite graph $G$ that encodes its structure.
%This graph has an edge $(a_i,x_j)$ if and only if the entry $M_{ij}$ of the matrix is nonzero.
The sparsity pattern of a matrix $M$ can be seen as the bipartite adjacency matrix of some bipartite graph $G$.
This bipartite graph fully encodes the structure of the matrix.
We assume here that the treewidth $\omega$ of $G$ is small.
Even though it is hard to determine the treewidth of a graph, there are many good heuristics and approximations that justify our assumption~\cite{bodlaender2008combinatorial}.
We show an algorithm to compute $\perm(M)$ in $\widetilde{O}(n\,2^{\omega})$ arithmetic operations.
In this paper, the notation $\widetilde{O}$ ignores polynomial factors in~$\omega$.
We note that the algorithm can be used over any commutative ring.

The permanent of a matrix can be generalized in several ways.
In particular, given a list of $n$ matrices of size $n\times n$, its \emph{mixed discriminant} generalizes both the permanent and the determinant \cite{Gurvits2005,Bapat1989}.
Our algorithm for the permanent extends in a natural way to compute the mixed discriminant.
The natural structure to represent the sparsity pattern in this case is a tripartite (i.e., $3$-colorable)  graph.
The running time of the resulting algorithm is $\widetilde{O}(n^2 + n\,3^{\omega})$, where $\omega$ is the treewidth of such graph.
In particular, this algorithm can compute the determinant of a matrix in the same time.

More generally, our methods extend to generalized determinants/permanents on tensors.
A special case of interest is the \emph{multidimensional permanent} \cite{Avgustinovich2010, Dow1987,Tichy2015}.
Another interesting case is the first Cayley \emph{hyperdeterminant}, also known as Pascal determinant, which is the simplest generalization of the determinant to higher dimensions \cite{Cayley1846,Luque2003,Barvinok1995}.
Note that unlike the determinant, the hyperdeterminant is $\sharpp$-hard, in particular because it contains mixed discriminants as a special case \cite{Gurvits2005,Hillar2013}.

Given a set of $n$ polytopes in $\R^n$, its mixed volume provides a geometric generalization of the permanent and the determinant \cite{Schneider2013}.
We focus on the special case of mixed volume of $n$ zonotopes.
Although there is no ``natural'' graph to represent the structure of a set of zonotopes, we associate to it a bipartite graph that, when the mixed volume restricts to a permanent, corresponds to the matrix graph described above.
This allows us to give a simple application for mixed volumes of zonotopes with few nonparallel edges.
Nevertheless, we show that mixed volumes remain hard to compute in the general case, even if this bipartite graph has treewidth~1 and the zonotopes have only 3 nonzero coordinates.

The diagram of Figure~\ref{fig:hassediagram} summarizes the scope of the paper.
It presents the main problems we consider, illustrating the relationships among them.
Concretely, an arrow from $A$ to $B$ indicates that $B$ is a special instance of $A$.
It also divides the problems according to its difficulty, with and without bounded treewidth assumptions.
In this paper we start from the simplest problems, i.e., permanents and determinants of matrices, moving upwards in the diagram.

  \begin{figure}[htb]
  \centering
  \includegraphics[scale=0.7]{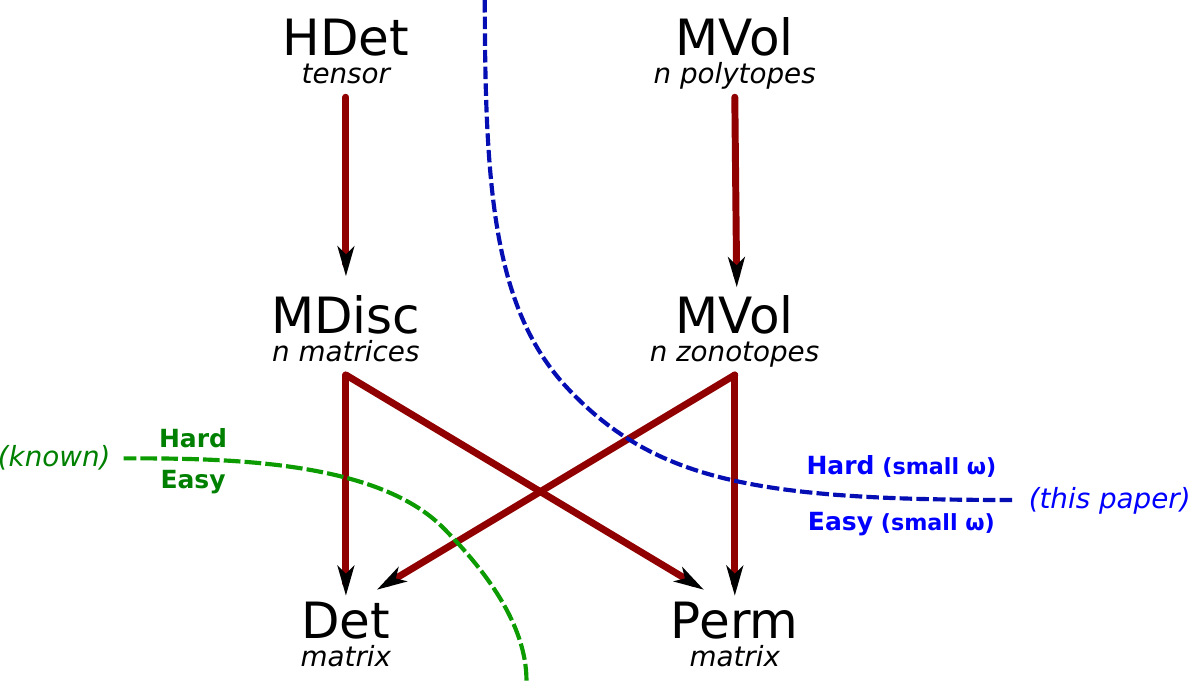}
  \caption{Diagram describing the complexity relations of computing: determinants, permanents, mixed discriminants, hyperdeterminants and mixed volumes.
}
  \label{fig:hassediagram}
  \end{figure}

The document is structured as follows.
In Section~\ref{s:treedecomposition} we review the concept of treewidth and tree decompositions.
We also present three graph abstractions of a sparse matrix.
Among these graphs is the bipartite graph $G$ described above, and a projection $G^X$ onto the column set.
In Section~\ref{s:treepermanent} we present a decomposition method, Algorithm~\ref{alg:colsperm}, that computes the permanent based on the graph $G^X$.
We first use graph $G^X$ because the decomposition algorithm is easier to explain in this case.
In Section~\ref{s:bipartite} we extend this method to work with the bipartite graph $G$, as shown in Algorithm~\ref{alg:bipartiteperm}.
We provide a Matlab implementation of this algorithm.
In Section~\ref{s:discriminant} we discuss the case of mixed discriminants, presenting a decomposition method for it.
We also treat the case of generalized determinants/permanents on tensors.
Finally, in Section~\ref{s:zonotope} we discuss the case of mixed volumes of zonotopes.

%%% Local Variables: 
%%% mode: latex
%%% TeX-master: "main"
%%% End: 

% Corrections Pablo -- June 1

\subsection*{Related work}

\subsubsection*{Permanents}
The best known to date method for exactly computing the permanent of a general matrix was given by Ryser and its complexity is $O(n\,2^{n})$ \cite{Ryser1963}.
%There are other exact general methods with the same complexity.
There are two main research trends on permanent computation:
%Research on permanent computation usually falls in two categories:
approximation algorithms and exact algorithms for structured matrices.
We briefly discuss related work in both of them.
Tree decomposition algorithms, which belong to the second group, will be presented afterwards.

We first mention some work on approximation schemes.
For arbitrary matrices, Gurvits gave a randomized polynomial time approximation, with error proportional to the operator norm \cite{Gurvits2005}.
%For arbitrary matrices, Gurvits gave a deterministic polynomial time approximation on the order of the operator norm of $M$ \cite{Gurvits2005}.
For nonnegative matrices there is a vast literature, see e.g., \cite{Vontobel2013} and the references therein.
Most remarkably, Jerrum et~al.\ gave a fully polynomial randomized approximation scheme (FPRAS) \cite{Jerrum2004}.
Recent work studies approximation schemes based on belief propagation, also for nonnegative matrices \cite{Vontobel2013,Watanabe2010}.

As for exact algorithms for structured matrices, different types of structure have been explored in the literature.
Fisher, Kasteleyn and Temperley gave a polynomial time algorithm for matrices whose associated bipartite graph is planar \cite{Kasteleyn1967,Temperley1961}.
%The case where the associated bipartite graph is planar was solved by Fisher, Kasteleyn and Temperley \cite{Kasteleyn1967,Temperley1961}.
Barvinok showed that the permanent is tractable when the rank is bounded \cite{Barvinok1996}.
The case of $0/1$ circulant matrices has also been considered \cite{Minc1987}, as well as sparse $0/1$ Toeplitz matrices \cite{Codenotti1997}.
Schwartz showed a $O(\log(n)2^{6w})$ algorithm for certain band, Toeplitz matrices, where $w$ is the bandwidth \cite{Schwartz2009}.
Temme and Wocjan showed a $O(n\,2^{3w^2})$ algorithm for a special type of band matrices \cite{Temme2012}.
Note that for arbitrary band matrices our algorithm is $\widetilde{O}(n\,2^{2w})$.

\subsubsection*{Permanents and treewidth}

Tree decomposition methods for permanent computation have been considered.
Courcelle et~al.\ first showed that the permanent can be computed efficiently if the treewidth is bounded \cite{Courcelle2001}, although their methods, based on the Feferman-Vaught-Shelah Theorem, do not lead to an implementable algorithm. % due to the underlying constants.
Later work of Flarup et~al.\ gives a  $O(n\,2^{O(\omega^2)})$ algorithm \cite{Flarup2007}.
This algorithm is extended in~\cite{Meer2011} to a wider class of matrices.
Note the strong dependency on the treewidth.
Furthermore, the graph abstraction used in the above methods, which is not the bipartite graph $G$, has two inconvenient features: its treewidth can be significantly larger than the one of $G$ (see Example~\ref{exmp:gridgraph}) and it is dependent on the specific order of the columns of the matrix (see Remark~\ref{rem:graphinvariance}).

%We note that the graph abstraction of the matrix used in the above methods is not the bipartite graph.
%The graph $G^s$ they use has an edge $ij$ whenever $M_{ij}\neq 0$ or $M_{ji}\neq 0$.
%This graph has two inconvenient features: it loses structure by symmetrizing the matrix, and it is dependent on the specific order of the columns (or rows) of the matrix, which should not matter for the permanent (see Remark~\ref{rem:graphinvariance}).
%In contrast, the bipartite graph $G$ we use is invariant under column (and row) permutations, and captures the full sparsity structure.
%Moreover, the treewidth of $G$ is always better than the one of $G^s$, as will be seen in Section~\ref{s:treedecomposition}.

Closer to this paper is the work of van Rooij et~al.~\cite{Rooij2009}.
They gave a $\widetilde{O}(n\,2^{\omega})$ decomposition algorithm for counting perfect matchings in a graph. %, out of which one could derive a similar algorithm for the permanent.
Counting perfect matchings is closely related to the permanent, and one could derive from their proof an analogous, but different, method for calculating the permanent.
Our algorithm could be seen as a variant of such method that is easier to extend to the higher dimensional problems we consider.

\subsubsection*{Mixed discriminants, mixed volumes, tensors}

The higher dimensional problems we study generalize the permanent of a matrix, and thus are $\sharpp$-hard in general.
As for the permanent, there are two natural relaxations: approximation algorithms and exact algorithms under special structure.
Approximation algorithms have been considered for mixed discriminants \cite{Barvinok1997,Gurvits2005}, mixed volumes \cite{Barvinok1997,Dyer1998} and multidimensional permanents \cite{Barvinok2011}.
As for exact algorithms under special structure, we are only aware of Gurvits' tractability result for mixed discriminants and the 4-hyperdeterminant under some bounded rank assumptions \cite{Gurvits2005}. 

To our knowledge this is the first paper that studies tree decomposition methods for mixed discriminants, mixed volumes and generalized determinants/permanents on tensors.
Related to this is a recent log-space algorithm for computing determinants under bounded treewidth assumptions~\cite{Balaji2015}.
Also related is the problem of partitioning a low treewidth graph into $k$-cliques, which is considered in~\cite{Rooij2009}.

%The higher dimensional problems we study generalize the permanent of a matrix, and thus are $\sharpp$-hard in general.
%There exist approximation algorithms for some of these problems: mixed discriminants \cite{Barvinok1997,Gurvits2005}, mixed volumes \cite{Barvinok1997,Dyer1998}, multidimensional permanents \cite{Barvinok2011}.
%As for exact algorithms under special structure, the literature is almost nonexistent.
%We remark Gurvits' tractability result for mixed discriminants and the 4-hyperdeterminant under some bounded rank assumptions \cite{Gurvits2005}. 
%To our knowledge this is the first paper that studies tree decomposition methods for mixed discriminants, mixed volumes and generalized determinants/permanents on tensors.
%Related to this is a recent log-space algorithm for computing the determinant of 0/1 matrices, assuming that $G^s$ has bounded treewidth \cite{Balaji2015}.

%%% Local Variables: 
%%% mode: latex
%%% TeX-master: "main"
%%% End: 

% Corrections Pablo -- July 7

\section{Tree decompositions and matrix  graphs }\label{s:treedecomposition}

In this section we review some basic facts regarding tree decompositions of a graph.
We also present three graphs that can be associated to a sparse matrix.

\subsection{Tree decompositions}

The notions of treewidth and tree decompositions are fundamental in many areas of computer science and applied mathematics \cite{Dechter2003,bodlaender2008combinatorial}.
Intuitively, the treewidth of a graph $G$ is a measure of how close it is to a tree.
A graph has treewidth~1 if and only if it is a forest, i.e., a disjoint union of trees.
The smaller the treewidth, the closer the graph is to a tree, and the easier it is to solve certain problems on it.
We note that a graph of treewidth $\omega$ has at most $n\omega$ edges, and thus treewidth imposes a sparsity constraint.
We give the formal definition now.

\begin{defn}\label{defn:treedecomposition}
Let $G$ be a graph with vertex set $X$.
A \emph{tree decomposition} of $G$ is a pair $(T,\chi)$, where $T$ is a rooted tree and $\chi: T\to \{0,1\}^X$ assigns some $\chi(t)\subset X$ to each node $t$ of $T$, that satisfies the following conditions.
\begin{enumerate}[ i.]
  \item The union of $\{\chi(t)\}_{t\in T}$ is the whole vertex set $X$.
  \item For every edge $(x_i,x_j)$ of $G$, there exists some node $t$ of $T$ with $x_i,x_j\in \chi(t)$.
  \item For every $x_i \in X$ the set $\{t:x_i\in \chi(t)\}$  forms a subtree of $T$.
\end{enumerate}
The sets $\chi(t)$ are usually referred to as \emph{bags}.
The \emph{width} $\omega$ of the decomposition is the size of the largest bag (minus one). 
The \emph{treewidth} of $G$ is the minimum width among all possible tree decompositions.
\end{defn}

Algorithms based on tree decompositions typically depend exponentially on the width of the decomposition, and polynomially on the number of nodes \cite{bodlaender2008combinatorial}.
Thus, given a graph $G$ it is desirable to obtain a tree decomposition of minimum width.
However, finding the treewidth is NP-hard \cite{Arnborg1987}.
The treewidth of some simple graphs are known:
for a tree is~$1$, 
for a cycle graph is~$2$,
for the $n\times n$ grid graph is~$n$,
for the complete graph $K_n$ is~$n-1$,
for the complete bipartite graph $K_{n,n}$ is~$n$.
For general graphs, there are good heuristics and approximation algorithms  \cite{bodlaender2008combinatorial}.

Tree decompositions are closely related to chordal graphs \cite{Blair1993,Dechter2003}.
Indeed, given a chordal graph $G$, we can construct a tree decomposition where the bags $\chi(t)$ correspond to its maximal cliques.
We remark that there are at most $n$ maximal cliques in a chordal graph.
In general, we can always assume that a tree decomposition has at most $n$ nodes.

The following is a simple property of tree decompositions.

\begin{lem}\label{thm:treedecompclique}
  Let $(T,\chi)$ be a tree decomposition of $G$.
  Then for any clique $Y$ of $G$ there is some node $t$ of $T$ with $Y\subset \chi(t)$.
\end{lem}
\begin{proof}
  For each vertex $y \in Y$, let $T_y$ denote the subtree of all bags containing $y$.
  Let $t_y\in T_y$ be the closest node to the root and let $d(t_y)$ denote the distance from $t_y$ to the root.
  Observe that if $d(t_y)\leq d(t_{y'})$ then $t_{y'}\in T_y$
  (otherwise, $T_y\cap T_{y'}=\emptyset$ and the edge $(y,y')$ would not belong to any bag).
  Let $t\in\{t_y\}_{y\in Y}$ be the farthest away from the root, i.e., $d(t_y)\leq d(t)$ for all $y\in Y$.
  It follows that $Y\subset \chi(t)$.
\end{proof}

\subsection{Matrix graphs}

The sparsity structure of a matrix, i.e., its pattern of nonzero entries, can be described in terms of a graph.
We consider here three possible graph abstractions of such sparsity structure, and we compare their treewidths.
%The sparsity structure of a matrix can be described in terms of a graph.
%We consider here three possible graphs and we compare their treewidth.

Let $M$ be a $n\times n$ matrix.
We will index the rows with a set $A = \{a_1,\ldots,a_n\}$ and the columns with a set $X= \{x_1,\ldots,x_n\}$.
We use subindices to index the coordinates of $M$, i.e., $M_{a,x}$ denotes the entry in the $a$-th row and $x$-th column.
Similarly, let $M_a$ be the $a$-th row of $M$.
We now present two (undirected) graphs that are usually associated to a sparse matrix.

\begin{defn}[Bipartite graph]\label{defn:graphbipartite}
  Let $M$ be a $n\times n$ matrix, let $A$ denote its set of rows, and let $X$ denote its set of columns.
  The \emph{bipartite graph} of $M$, denoted as $G(M)$,  has vertices $A\cup X$, and there is an edge $(a,x)$ if  $M_{a,x}$ is nonzero.
\end{defn}

\begin{defn}(Symmetrized graph)\label{defn:graphsymmetric}
  Let $M$ be a $n\times n$ matrix, let $A$ denote its set of rows, and let $X$ denote its set of columns.
  The \emph{symmetrized graph} of $M$, denoted as $G^s(M)$, has vertices $1,\ldots, n$ and has an edge $(i,j)$ if $M_{a_i,x_j}\neq 0$ or $M_{a_j,x_i}\neq 0$.
\end{defn}
\begin{rem}
  Note that $G^s(M)$ is the adjacency graph of the symmetric matrix $M+M^T$, assuming no terms cancel out.
\end{rem}

\begin{rem}[Permutation invariance]\label{rem:graphinvariance}
  Note that the permanent of a matrix is invariant under independent row and column permutations.
  The bipartite graph $G$ preserves this invariance.
  On the other hand, the symmetrized graph $G^s$ is only invariant under simultaneous row and column permutations.
\end{rem}

The bipartite graph $G$ is our main object of study in this paper.
Let $\omega:=\tw(G)$ be the treewidth of $G$ and $\omega_s:=\tw(G^s)$ the one of $G^s$.
Tree decomposition methods based on graph $G^s$ have been studied before 
%The tractability of the permanent when $\omega_s$ is bounded has been studied 
\cite{Courcelle2001, Flarup2007, Meer2011, Balaji2015}.
We claim that graph $G$ is a better abstraction for the purpose of permanent computation.
In particular, $G$ preserves the permutation invariance of the permanent as stated above.
Furthermore, $\omega_s$ can be much larger than $\omega$ as shown in the following example.

\begin{exmp}[Two nonzero entries per row]\label{exmp:gridgraph}
  \begin{figure}[htb]
  \begin{minipage}{.5\textwidth}
    {\footnotesize \begin{align*}
    \begin{blockarray}{cccccccccc}
    & x_1\!\!\! & x_2\!\!\! & x_3\!\!\! & x_4\!\!\! & x_5\!\!\! & x_6\!\!\! & x_7\!\!\! & x_8\!\!\! & x_9 \\
    \begin{block}{c[ccccccccc]}
      a_1\mbox{ } & 0 & 1 & 0 & 2 & 0 & 0 & 0 & 0 & 0\\ 
      a_2\mbox{ } & 0 & 0 & 1 & 0 & 2 & 0 & 0 & 0 & 0\\ 
      a_3\mbox{ } & 0 & 0 & 3 & 0 & 0 & 2 & 0 & 0 & 0\\ 
      a_4\mbox{ } & 0 & 0 & 0 & 0 & 1 & 0 & 2 & 0 & 0\\ 
      a_5\mbox{ } & 0 & 0 & 0 & 0 & 0 & 1 & 0 & 2 & 0\\ 
      a_6\mbox{ } & 0 & 3 & 0 & 0 & 0 & 0 & 0 & 0 & 2\\ 
      a_7\mbox{ } & 0 & 0 & 0 & 0 & 0 & 0 & 3 & 1 & 0\\ 
      a_8\mbox{ } & 0 & 0 & 0 & 3 & 0 & 0 & 0 & 0 & 1\\ 
      a_9\mbox{ } & 3 & 0 & 0 & 0 & 0 & 0 & 0 & 0 & 0\\
    \end{block}
    \end{blockarray}
  \end{align*}}
  \end{minipage}%
  \begin{minipage}{.5\textwidth}
  \centering
  \includegraphics[scale=0.4]{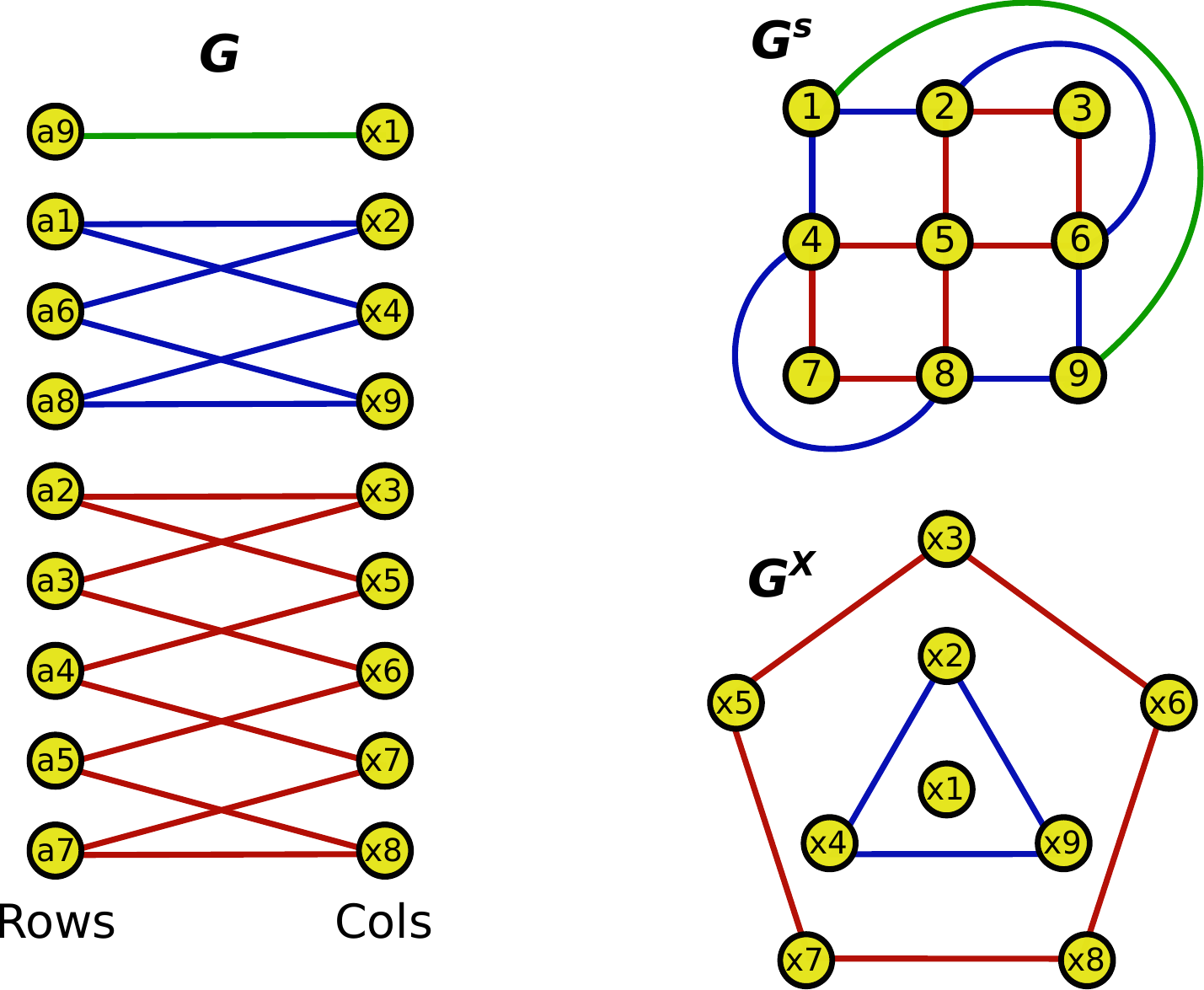}
  \end{minipage}
  \caption{Graph abstractions of a matrix: bipartite graph $G$, symmetrized graph $G^s$ and column graph $G^X$.}
  \label{fig:gridgraph}
  \end{figure}

  Let $M$ be a matrix with at most two nonzero entries per row.
  We claim that for all nontrivial cases the bipartite graph $G$ has treewidth $\omega\leq 2$.
  Let $G_0$ be a connected component, and let $n_0$ be its number of row vertices.
  In order for $G_0$ to have a perfect matching, it must have as many row vertices as column vertices.
  Note also that $G_0$ has at most $2n_0$ edges because the row degrees are at most $2$.
  Then $G_0$ is a connected graph with $2n_0$ vertices and at most $2n_0$ edges, so it has at most one cycle.
  It follows that $\omega\leq 2$.

  On the other hand, we will see that $\omega_s$ is unbounded.
  Let $n= m^2$ and consider the matrix $M$ whose nonzero entries are
  \begin{align*}
    M_{a_{i},\,x_{i+1}} &= 1, &\mbox{if $m$ does not divide $i$} \\
    M_{a_{i},\,x_{i+m}} &= 2, &\mbox{if $1\leq i\leq n-m$ \qquad} \\
    M_{a_{(m-i+1)m},\,x_{i}} &= 3, &\mbox{if $1\leq i\leq m$\qquad\qquad\,} \\
    M_{a_{n-i},\,x_{im+1}} &= 3, &\mbox{if $1\leq i< m$\qquad\qquad\,} 
  \end{align*}
  The corresponding graph $G^s$ contains the grid graph, and thus $\omega_s\geq \sqrt{n}$.
  The case $n=9$ is shown in Figure~\ref{fig:gridgraph}.
  \begin{align*}
  \end{align*}
\end{exmp}

The following example shows that the treewidth of $G$ is always better than the treewidth of $G^s$.

\begin{exmp}[$G$ is ``better'' than $G^s$]\label{exmp:GbetterGs}
  Let's see that given a tree decomposition of $G^s$ of width $\omega_s$ we can form a tree decomposition of $G$ of width $2\omega_s$.
  Let $(T,\iota)$ be a tree decomposition of $G^s$, where $\iota(t)\subset \{1,\ldots,n\}$.
  Let $\mu:T\to \{0,1\}^{A\cup X}$, be such that $\mu(t) = \{a_i: i\in \iota(t)\} \cup \{x_i: i\in \iota(t)\}$.
  Then $(T,\mu)$ is a decomposition of $G$ of width $2\omega_s$.
  On the contrary, for a fixed $\omega$ the treewidth of $G^s$ is unbounded as seen in Example~\ref{exmp:gridgraph}.
\end{exmp}

We now introduce a third graph $G^X$ that we can associate to matrix $M$.

\begin{defn}(Column graph)\label{defn:graphrow}
  Let $M$ be a $n\times n$ matrix, let $A$ denote its set of rows and let $X$ denote its set of columns.
  For any $a\in A$ let $X(a)$ denote the set of nonzero components of row $M_a$.
  The \emph{column graph} $G^X(M)$ has $X$ as its vertex set, and for each $a\in A$ we put a clique in $X(a)$.
  Equivalently, there is an edge $(x_i, x_j)$ if there is some $a\in A$ such that $x_i,x_j \in X(a)$.
\end{defn}

%\begin{rem}
%  Observe that $G(M)$ and $G^s(M)$ are invariant under matrix transposition, whereas $G^X(M)$ is not.
%  Also note that $G(M)$ and $G^X(M)$ are invariant under row and column permutations, whereas $G^s(M)$ is not.
%\end{rem}

Graph $G^X$ can be seen as a projection of $G$ onto the column set $X$.
We show in the following example that $\omega\leq \omega_X+1$, where $\omega_X:=\tw(G^X)$.

\begin{exmp}[$G$ is ``better'' than $G^X$]\label{exmp:GbetterGX}
  Let $(T,\chi)$ be a tree decomposition of $G^{X}$ of width $\omega_X$.
  For each row $a\in A$ we associate to it a unique node $t\in T$ such that $X(a)\subset \chi(t)$.
  This assignment can be made because of Lemma~\ref{thm:treedecompclique}.
  For some $t\in T$, let $a^t_1,a^t_2,\ldots,a^t_k$ be all rows that are assigned to $t$.
  Let's replace node $t$ of $T$ with a path $t_1,t_2\ldots,t_k$, and let 
  $\mu(t_j)= \chi(t)\cup \{a^t_j\}$ for $j=1,\ldots,k$.
  The nodes previously connected to $t$ can be linked to any of the new nodes.
  By doing this for every $t\in T$, we obtain a tree decomposition $(T,\mu)$ of $G$ of width $\omega_X +1$.

  On the other hand, for a fixed $\omega$ the treewidth of $G^X$ is unbounded.
  For instance, consider the matrix $M$ whose only nonzero entries are: $M_{a_i,x_i}=1$ (diagonal) and $M_{a_1,x_i}=1$ (first row) for all $i$.
  In this case $G^X$ is the complete graph ($\omega_X = n-1$) and $G$ is a tree ($\omega = 1$).
\end{exmp}

The reason why we consider the graph $G^X$ is that we can give a very simple algorithm for the permanent based on it.
We present this algorithm in Section~\ref{s:treepermanent}, and then extend it to $G$ in Section~\ref{s:bipartite}.

To conclude this section, let's see how the three graphs $G,G^s,G^X$ can capture the special case of a band matrix.

\begin{exmp}[Band matrix]
  Let $w_1,w_2\in \Z_{>0}$ and let $M$ be such that $M_{a_i,x_j}=0$ if either $i-j > w_1$ or $j-i> w_2$.
  Let $T$ be a path $t_1,t_2,t_3,\ldots$.
  Optimal decompositions $(T,\chi), (T,\iota), (T,\mu)$ for the graphs $G^X,G^s,G$ respectively, are given by 
  \begin{gather*}
  \chi(t_i) := \{x_{i},x_{i+1},\ldots, x_{i+w_1+w_2}\}, \qquad\qquad
  \iota(t_i) := \{i,i+1,\ldots, i+\max\{w_1,w_2\}\}, \\
  \mu(t_{2i-1}):= \{a_{i+w_1},x_{i},x_{i+1},\ldots, x_{i+w_1+w_2-1}\}, \quad
  \mu(t_{2i}):= \{a_{i+w_1},x_{i+1},x_{i+1},\ldots, x_{i+w_1+w_2}\} 
  \end{gather*}
  The widths of these decompositions are $\omega= \omega_X=w_1+w_2$ and $\omega_s = \max\{w_1,w_2\}$.
\end{exmp}

%%% Local Variables: 
%%% mode: latex
%%% TeX-master: "main"
%%% End: 

% Corrections Pablo -- July 7

\section{Column decompositions}\label{s:treepermanent}

\begin{notation}
  For sets $Y,Y_1,Y_2$, let $Y = Y_1\sqcup Y_2$ denote a set partition, i.e., $Y = Y_1\cup Y_2$ and $Y_1\cap Y_2 = \emptyset$.
\end{notation}

In this section, we develop an algorithm to compute the permanent based on a tree decomposition of the column graph $G^X$ (see Definition~\ref{defn:graphrow}).
For this section only, we denote the treewidth of $G^X$ by $\omega$ .
We will show that we can compute $\perm(M)$ in $\widetilde{O}(n\,4^{\omega})$.
Recall that the notation $\widetilde{O}$ ignores polynomial factors in $\omega$.

As before, $A$ denotes the row set and $X$ the column set.
We use subindices to index the coordinates of $M$, i.e., $M_{a,x}$ denotes the element in row $a$ and column $x$.
The following example illustrates the methodology we follow.

\begin{exmp}\label{exmp:blockmatrix}
  \begin{figure}[htb]
  \begin{minipage}{.6\textwidth}
    {\footnotesize \begin{align*}
    M =\begin{bmatrix}
M_{a_1,x_1} & 0 & M_{a_1,x_3} & M_{a_1,x_4}& 0\\ 
M_{a_2,x_1} & 0 & M_{a_2,x_3} & M_{a_2,x_4}& 0\\ 
0 & M_{a_3,x_2} & M_{a_3,x_3} & M_{a_3,x_4}& 0\\ 
0 & M_{a_4,x_2} & M_{a_4,x_3} & 0& M_{a_4,x_5}\\
0 & M_{a_5,x_2} & M_{a_5,x_3} & 0& M_{a_5,x_5}
    \end{bmatrix}
  \end{align*}}
  \end{minipage}%
  \begin{minipage}{.4\textwidth}
  \centering
  \includegraphics[scale=0.45]{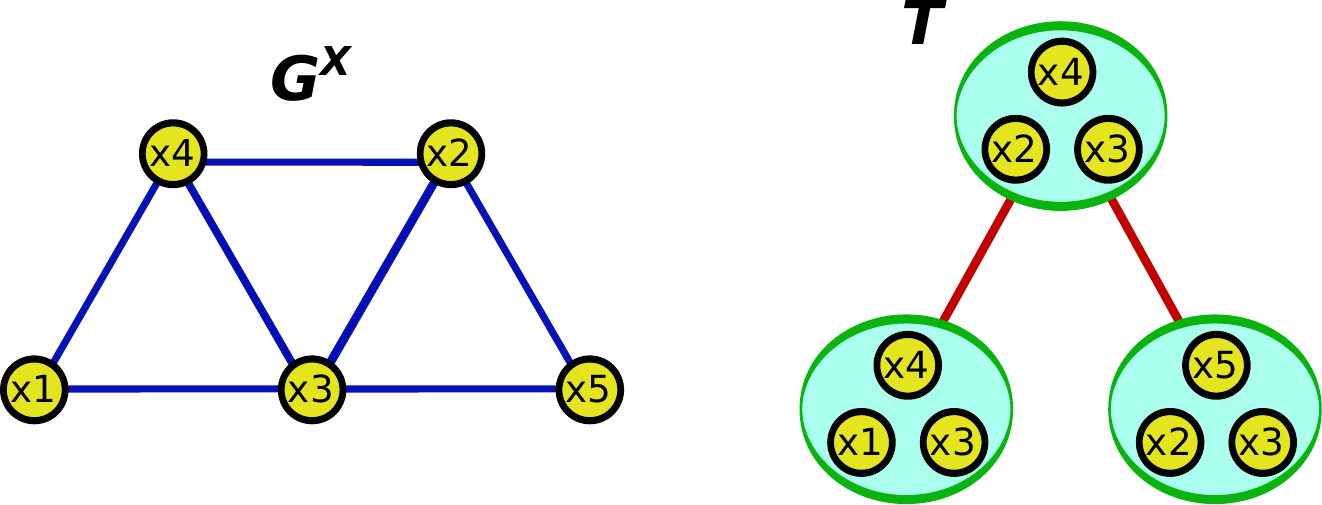}
  \end{minipage}
  \vspace{-10 pt}
  \caption{Matrix $M$, its column graph $G^X$ and a tree decomposition $T$.}
  \label{fig:triangles}
  \end{figure}
  Consider the $5\times 5$ matrix $M$ of Figure~\ref{fig:triangles} and the following partition of its rows:
  \begin{align*}
    A = \{a_1,a_2\} \sqcup \{a_3\} \sqcup \{a_4,a_5\}.
  \end{align*}
  There is a simple expansion of $\perm(M)$ in terms of this partition:
  \begin{equation*}
  \begin{aligned}
    \perm(M) &= 
    \perm\left(
\begin{bmatrix} M_{a_1,x_1} & M_{a_1,x_3}\\ M_{a_2,x_1} & M_{a_2,x_3} \end{bmatrix}
    \right)
    \perm\left(
\begin{bmatrix} M_{a_3,x_4} \end{bmatrix}
    \right)
    \perm\left(
\begin{bmatrix} M_{a_4,x_2} & M_{a_4,x_5}\\ M_{a_5,x_2} & M_{a_5,x_5} \end{bmatrix}
    \right)\\
    &+\perm\left(
\begin{bmatrix} M_{a_1,x_1} & M_{a_1,x_4}\\ M_{a_2,x_1} & M_{a_2,x_4} \end{bmatrix}
    \right)
    \perm\left(
\begin{bmatrix} M_{a_3,x_3} \end{bmatrix}
    \right)
    \perm\left(
\begin{bmatrix} M_{a_4,x_2} & M_{a_4,x_5}\\ M_{a_5,x_2} & M_{a_5,x_5} \end{bmatrix}
    \right)\\
    &+\perm\left(
\begin{bmatrix} M_{a_1,x_1} & M_{a_1,x_4}\\ M_{a_2,x_1} & M_{a_2,x_4} \end{bmatrix}
    \right)
    \perm\left(
\begin{bmatrix} M_{a_3,x_2} \end{bmatrix}
    \right)
    \perm\left(
\begin{bmatrix} M_{a_4,x_3} & M_{a_4,x_5}\\ M_{a_5,x_3} & M_{a_5,x_5} \end{bmatrix}
    \right).
  \end{aligned}
\end{equation*}
  This expansion implies that to compute $\perm(M)$ we just need to evaluate two $2\times 2$ permanents corresponding to $\{a_1,a_2\}$, three $1\times 1$ permanents corresponding to $\{a_3\}$, and two $2\times 2$ permanents corresponding to $\{a_4,a_5\}$.
  This requires only $14$ multiplications, compared to $4\times 5! = 480$ multiplications using the definition.
  The reason why this formula exists is because the column graph $G^X$ of matrix $M$ has a simple tree decomposition, which is shown in Figure~\ref{fig:triangles}.
    %$\{x_1,x_3,x_4\}, \{x_2,x_3,x_4\}, \{x_2,x_3,x_5\}.$
\end{exmp}

As in the example above, we can always obtain an expansion of $\perm(M)$ using a tree decomposition of graph $G^X$.
By carefully evaluating this formula we will obtain a dynamic programming method to compute $\perm(M)$.

\subsection{Partial permanent}\label{s:partialperm}

In our notation, the permanent of $M$ can be expressed as
\begin{align*}
  \perm(M) = \sum_\pi \prod_{a \in A} M_{a,\pi(a)}
\end{align*}
where the sum is over all bijections $\pi: A \to X$.
For a given row $a$, let $X(a)$ denote the column coordinates where it is nonzero.
We will refer to a bijection $\pi$ as a \emph{matching} if $\pi(a)\in X(a)$, i.e., $M_{a,\pi(a)}\neq 0$, for all $a\in A$.
Then we can restrict the above sum to be over all matchings.

We consider a tree decomposition $(T,\chi)$ of the column graph $G^X$.
Note that by construction of $G^X$ then $X(a)$ is a clique for any $a\in A$.
Thus, Lemma~\ref{thm:treedecompclique} says that we can assign each row $a\in A$ to some node $t$, such that $X(a)\subset \chi(t)$.
From now, we \emph{fix an assignment} of each $a$ to a unique node.
Let $A_{t}$ denote the rows that are assigned to node $t$.
Thus, we have the partition $A = \bigsqcup_{t\in T}A_{t}$.

For a node $t\in T$, we denote the subtree of $T$ rooted in $t$ by $T_t$.
Let $\chi(T_t)$ be the union of $\chi(t')$ over all $t'\in T_t$, and similarly let $A_{T_t}$ be the union of $A_{t'}$ over all $t'\in T_t$.
For instance, for the root node we have $A_{T_{\mathit{root}}} = A$ and $\chi(T_{\mathit{root}}) = X$.

For a fixed matrix $M$ and some sets $D\subset A$ and $Y\subset X$ we denote 
\begin{align}\label{eq:defnpartialperm}
  \myperm(D,Y) := \sum_{\pi} \prod_{a \in D} M_{a,\pi(a)} 
\end{align}
where the sum is over all matchings $\pi: D\to Y$.
Equivalently, it is the permanent of the submatrix of $M$ corresponding to such rows and columns.
Clearly, this only makes sense if $\card{D}=\card{Y}$, and otherwise we can define $\myperm(D,Y)=0$.
We refer to $\myperm(D,Y)$ as a \emph{partial permanent}.

\subsection{Decomposition algorithm for the permanent}

Algorithm~\ref{alg:colsperm} presents our dynamic programming method to compute $\perm(M)$.
We will explain and derive this algorithm in the following sections.
For each node $t$ of the tree, the algorithm computes a table $P_t$ indexed by subsets $\bar{Y}$ of $\chi(t)$.
It starts from the leaves of the tree and recursively computes the tables of all nodes following a topological ordering.
The permanent of $M$ is found in the table corresponding to the root.

\begin{algorithm}
  \caption{Permanent with column decomposition}
  \label{alg:colsperm}
  \begin{algorithmic}[1]
    \Require{Matrix $M$ and tree decomposition $(T,\chi)$ of column graph $G^X(M)$}
    \Ensure{Permanent of $M$}
    \Procedure{ColsPerm}{$M,T,\chi$}
    \State assign each $a\in A$ to some $t$ with $X(a)\subset \chi(t)$
    \State $\mathit{order} := $ topological ordering of $T$ starting from its leaves
    \For {$t$ in $\mathit{order}$}
    \State $Q_t:=$ \Call{SubPerms}{$t,M$}
    \If{$t$ is a leaf}
    \State $P_t := Q_t$
    \Else
    \State ${c_1},\ldots,{c_k}:=$ children of $t$
    \State $P_t := $\Call{EvalRecursion}{$t,Q_t,P_{c_1},\ldots,P_{c_k}$}
    \EndIf
    \EndFor
    \State \Return $P_{\mathit{root}}(\chi(\mathit{root}))$
    \EndProcedure
    \Procedure{SubPerms}{$t,M$}
    \State $A_t:=$ rows assigned to node $t$
    \State $Q_t(\bar{Y}) := \myperm(A_t,\bar{Y})$ \; \algorithmicforall\ $ \bar{Y}\subset\chi(t)$
    \EndProcedure
    \Procedure{EvalRecursion}{$t,Q_t,P_{c_1},\ldots,P_{c_k}$}
    \For {${c_j}$ child of $t$}
    \State $\Delta_j:= \chi(c_j)\setminus \chi(t)$,\quad $\Lambda_j:= \chi(c_j)\cap \chi(t)$
    \State $Q_{c_j}(\bar{Y}) := P_{c_j}(\bar{Y}\cup \Delta_j)$ \; \algorithmicforall\ $ \bar{Y}\subset \Lambda_j$
    \EndFor
    \State $P_t:=$ \Call{SubsetConvolution}{$Q_t,Q_{c_1},\ldots,Q_{c_k}$}
    \EndProcedure
    \Procedure{SubsetConvolution}{$P_0,P_1,\ldots,P_k$}
    \State $P(\bar{Y}) := \displaystyle\sum\limits_{\bar{Y}_0\sqcup\cdots\sqcup \bar{Y}_k= \bar{Y}}P_0(\bar{Y}_0) \,P_1(\bar{Y}_1)\cdots P_k(\bar{Y}_k)$  
%\;\; \algorithmicforall\ $ \bar{Y}$
    \EndProcedure
  \end{algorithmic}
\end{algorithm}

Algorithm~\ref{alg:colsperm} has two main routines: 
\begin{itemize}
  \item For any node $t$, \Call{SubPerms}{} computes a table $Q_t$ with the permanents of all submatrices corresponding to $t$.
    (See Section~\ref{s:subperms}).
  \item \Call{EvalRecursion}{} computes table $P_t$ of an internal node $t$, by combining table $Q_t$ with the tables $P_{c_1},\ldots,P_{c_k}$ of the node's children.
    (See Section~\ref{s:colsrecursion}).
\end{itemize}

The values $P_t(\bar{Y})$ that we compute correspond to a partial permanent of the matrix, as we explain now.
Consider the collection
  \begin{align*}
    S := \{Y:\chi(T_t) \setminus  \chi(t)\subset Y \subset \chi(T_t)\}.
  \end{align*}
Observe that $Y\in S$ is completely determined by $Y\cap \chi(t)$.
Therefore, if we let $\bar{Y}:=Y\cap \chi(t)$, there is a one to one correspondence between $S$, and the collection $\bar{S} := \{\bar{Y}: \bar{Y}\subset \chi(t)\}$.
Then the partial permanents that we are interested in are
\begin{align*}
  P_t(\bar{Y}) := \myperm(A_{T_t},Y) = \myperm(A_{T_t},\bar{Y}\cup (\chi(T_t)\setminus \chi(t))).
\end{align*}
The reason why we index table $P_t$ with $\bar{Y}$ instead of $Y$, is that in this way it becomes clearer that the recursion formula is actually a subset convolution.

Observe that the permanent of $M$ is indeed computed in Algorithm~\ref{alg:colsperm}, as for the root node we have $P_{\mathit{root}}(\chi(\mathit{root}))=\myperm(A,X)=\perm(M)$.
Also note that for a leaf node we have $P_t(\bar{Y})= \myperm(A_{t},\bar{Y})= Q_t(\bar{Y})$.

\begin{exmp}
  Consider the matrix $M$ and tree decomposition $T$ of Figure~\ref{fig:triangles}.
  Let $t_1,t_2,t_3$ be the nodes of $T$, where the central node $t_2$ is the root.
  We show the tables computed by Algorithm~\ref{alg:colsperm}.
  The tables $Q_t$ with the permanents of all submatrices are:
  \begin{align*}
    Q_{t_1}(\{x,y\}) &= \myperm(\{a_1,a_2\},\{x,y\} ),  &\mbox{for }x,y\in \chi(t_1) = \{x_1,x_3,x_4\}\\
    Q_{t_3}(\{x,y\}) &= \myperm(\{a_4,a_5\},\{x,y\} ),  &\mbox{for }x,y\in \chi(t_3) = \{x_2,x_3,x_5\}\\
    Q_{t_2}(\{x\})   &= \myperm(\{a_3\}    ,\{x\}   ),  &\mbox{for }x  \in \chi(t_2) = \{x_2,x_3,x_4\}
  \end{align*}
  We now show the final tables $P_t$ for each node.
  For the leaves $t_1,t_3$ we have $P_{t_1} = Q_{t_1}$, $P_{t_3}=Q_{t_3}$.
  As for the root $t_2$, the recursion is:
  \begin{multline*}
    P_{t_2}(\{x_2,x_3,x_4\}) = Q_{t_2}(\{x_4\})P_{t_1}(\{x_1,x_3\})P_{t_3}(\{x_2,x_5\}) \,+\\ Q_{t_2}(\{x_3\})P_{t_1}(\{x_1,x_4\})P_{t_3}(\{x_2,x_5\})  + Q_{t_2}(\{x_2\})P_{t_1}(\{x_1,x_4\})P_{t_3}(\{x_3,x_5\}).
  \end{multline*}
  Note that this recursion matches the permanent expansion in Example~\ref{exmp:blockmatrix}.
\end{exmp}

In the following sections we explain the two main routines of Algorithm~\ref{alg:colsperm}, i.e., \Call{SubPerms}{} and \Call{EvalRecursion}{}, obtaining complexity bounds for them.

\subsection{Permanent of all submatrices}\label{s:subperms}

Let $M_0$ be a rectangular matrix with row set $A_0$ and column set $X_0$.
As a part of our algorithm, which can be seen as the base case, we require a good method to compute the permanents of all submatrices of $M_0$.
In other words, we want to obtain the partial permanents $\myperm(D,Y)$ for all pairs $(D,Y)$.
We can do this in a very simple way using an expansion by minors.
The following lemma explains such procedure and gives its running time.

\begin{lem}\label{thm:subperms}
  Let $M_0$ be a matrix of dimensions $n_1\times n_2$.
  Let $A_0$ denote its row set, $X_0$ its column set and let $S=\{(D,Y)\subset A_0\times X_0:\, \card{D}=\card{Y}\}$.
  We can compute $\myperm(D,Y)$ for all $(D,Y)\in S$ in $O(n_{max}^2\,2^{n_1+n_2})$, where $n_{max}=\max\{n_1,n_2\}$.
\end{lem}
\begin{proof}
  Let
  $S_i=\{(D,Y):\, \card{D}=\card{Y}= i\}$
  for $1\leq i\leq \min\{n_1,n_2\}$.
  We use an expansion by minors to compute $\myperm(D,Y)$ for $(D,Y)\in S_i$, using the values of $S_{i-1}$.
  Let $a_0$ be the first element in $D$, then
  \begin{align*}
    \myperm(D,Y) = \sum_{x\in Y}\, M_{a_0,x}\;\myperm(D\setminus a_0, Y\setminus x).
  \end{align*}
  Thus, for each $(D,Y)$, we loop over at most $n_2$ elements, and for each we need $O(n_{max})$ to find the sets $D\setminus a_0$ and $Y\setminus x$.
  The result follows.
\end{proof}

\subsection{Recursion formula}\label{s:colsrecursion}

The heart of Algorithm~\ref{alg:colsperm} is given by the recursion formula used, i.e., the procedure to obtain table $P_t$ of node $t$ from the tables of its children.
This recursion formula is given in the following lemma.
%The following lemma has the basic recursion formula that determines our algorithm.

\begin{lem}\label{thm:recursepermcols}
  Let $M$ be a matrix with associated column graph $G^X$. 
  Let $(T,\chi)$ be a tree decomposition of $G^X$.
  Let $t$ be an internal node of $T$, and let $Y$ be such that 
  \begin{align}
     \chi(T_t) \setminus  \chi(t)\subset Y \subset \chi(T_t), \;\;\;\;\;\;\;
     \card{Y} = \card{A_{T_t}}
    \label{eq:validrangeR}
  \end{align}
  Let ${c_1},\ldots,{c_k}$ be the children of $t$. Then 
  \begin{align}\label{eq:recursepermcols}
    \myperm(A_{T_t}, Y) &= \sum_{\mathcal{Y}} \myperm(A_{t},Y_t) \prod_{j=1}^k \myperm(A_{T_{c_j}}, Y_{c_j} ) 
  \end{align}
  where $\myperm(\cdot,\cdot)$ is as in \eqref{eq:defnpartialperm} and the sum is over all $\mathcal{Y} = (Y_t,Y_{c_1},\ldots,Y_{c_k})$ such that:
  \begin{subequations} \label{eq:partitionrangepi}
  \begin{align}\label{eq:partitionrangepi1}
   Y &=  Y_t \sqcup (Y_{c_1} \sqcup \cdots \sqcup Y_{c_k}) \\
   \chi(T_{c_j}) \setminus  \chi(t) &\subset Y_{c_j} \subset \chi(T_{c_j})  \;\;\;\;\;\;\;\;\;\;\;\;\;
   Y_t\subset \chi(t).
  \end{align}
  \end{subequations}
\end{lem}

\begin{proof}
  Observe that $A_{T_t}$ can be partitioned as 
  \begin{align*}
    A_{T_t} = A_{t}\sqcup (A_{T_{c_1}} \sqcup \cdots \sqcup A_{T_{c_k}}).
  \end{align*}
  Let $\pi:A_{T_t}\to Y$ be a matching.
  Let $c$ be a child of $t$ and let $\pi_c: A_{T_c} \to Y$ be the restriction of $\pi$ to $A_{T_c}$.
  Let $Y_t := \pi(A_{t}) \subset \chi(t)$ be the range of $\pi$ restricted to $A_{t}$, and let $Y_c$ be the range of $\pi_c$.
  As $\pi$ is injective, then equation~\eqref{eq:partitionrangepi1} holds.
  Observe also that $Y_c=\pi(A_{T_c})\subset \chi(T_c)$.
  Note now that if $x \in \chi(T_c) \setminus  \chi(t)$, then it is in the range of $\pi$.
  However, as $x\notin \chi(t)$ then $x\notin \chi(T_{c'})$ for any other child ${c'}$, and thus $x$ has to be in the range of $\pi_c$.
  Thus, the range of $\pi_c$, i.e., $Y_c$, contains $\chi(T_c)\setminus \chi(t)$.

  Therefore, for any  matching $\pi:A_{T_t}\to Y$ and for any child $c$, $\pi$ induces a matching from $A_{T_c}$ to some $Y_c$ that satisfy equations~\eqref{eq:partitionrangepi}.
  On the other hand, given $Y_t,Y_{c_1},\ldots$ satisfying~\eqref{eq:partitionrangepi}  and matchings $\pi_t,\pi_{c_1},\ldots$ on $A_{t},A_{T_{c_1}},\ldots$ with such ranges, we can merge them into a function on $A_{T_t}$.
  Observe that \eqref{eq:partitionrangepi} ensures that the ranges of these matchings are disjoint and their union is $Y$.
  We conclude that
  %These remarks imply equation~\eqref{eq:recursepermcols}. 
  \begin{align*}
    \myperm(A_{T_t},Y) 
    &= \sum_{\pi:A_{T_t}\to Y}\prod_{a}M_{a,\pi(a)}\\
    &= \sum_{\mathcal{Y}}\sum_{\substack{\pi_t: A_{t}\to Y_t\\ \pi_{c}:A_{T_c}\to Y_c}}\left(\prod_{a_t}M_{a_t,\pi_t(a_t)}\right)\prod_{c_j}\left(\prod_{a_{c}}M_{a_{c},\pi_{c_j}(a_c)}\right)\\
    &= \sum_{\mathcal{Y}} \myperm(A_{t},Y_t) \prod_{c_j} \myperm(A_{T_{c_j}}, Y_{c_j} ).
  \end{align*}
\end{proof}

At first sight, the recursion of equation~\eqref{eq:recursepermcols} looks difficult to evaluate.
It turns out that this formula is a subset convolution and thus it can be computed efficiently using the algorithm from~\cite{Bjoerklund2007}, as explained in the following lemma.
%We now show an efficient way to do it based on the fast subset convolution from~\cite{Bjoerklund2007}.
We follow this approach in method \Call{EvalRecursion}{} of Algorithm~\ref{alg:colsperm}.

\begin{lem}\label{thm:countpartitions}
  Given the values of the partial permanents $\myperm(A_t,Y_t)$ and $\myperm(A_{T_{c_j}},Y_{c_j})$, we can evaluate equation~\eqref{eq:recursepermcols} for all $Y$ satisfying~\eqref{eq:validrangeR} in $\widetilde{O}(k\,2^{\omega})$.
\end{lem}
\begin{proof}
  Let $\bar{Y}:= Y\cap \chi(t)$, $\bar{Y}_t:=Y_t$, $\bar{Y}_{c_j}:=Y_{c_j}\cap \chi(t)$, and let
  \begin{gather*}
    \Delta_{t}:=\chi(T_{t})\setminus\chi(t) \qquad \Delta_{c_j}:=\chi(T_{c_j})\setminus\chi(t)\\
    P_t(\bar{Y}) := \myperm(A_{T_t},\bar{Y}\cup\Delta_t),\quad 
    Q_t(\bar{Y}_t) := \myperm(A_t,\bar{Y}_t),\quad
    Q_{c_j}(\bar{Y}_{c_j}) := \myperm(A_{T_{c_j}},\bar{Y}_{c_j}\cup \Delta_{c_j})
  \end{gather*}
  Then equation~\eqref{eq:recursepermcols} can be rewritten as
  \begin{align}\label{eq:recurseconv}
    P_t(\bar{Y}) &= \sum_{\bar{Y}_t\sqcup\bar{Y}_{c_1}\sqcup\cdots\sqcup\bar{Y}_{c_k}=\bar{Y}} Q_t(\bar{Y}_t) \prod_{j=1}^k Q_{c_j}(\bar{Y}_{c_j} ) 
  \end{align}
  where $\bar{Y}_t\subset \chi(t)$ and $\bar{Y}_{c_j} \subset \chi(c_j)\cap \chi(t)$.
  Equation~\eqref{eq:recurseconv} is a subset convolution over the subsets of $\chi(t)$.
  Therefore, it can be evaluated in $O(kw^2\,2^{{w}})$, where $w=|\chi(t)|$, using the algorithm from~\cite{Bjoerklund2007}.
\end{proof}

The following theorem gives the running time of Algorithm~\ref{alg:colsperm}, proving that we can efficiently compute the permanent given a tree decomposition of $G^X$ of small width.

\begin{thm}\label{thm:treewidthpermanent}
  Let $M$ be a matrix with associated column graph $G^X$. 
  Let $(T,\chi)$ be a tree decomposition of $G^X$ of width $\omega$.
  Then we can compute $\perm(M)$ in $\widetilde{O}(n\,4^{\omega})$.
\end{thm}
\begin{proof}
  Let $t$ be some node in $T$.
  We compute $\myperm(A_{T_t},Y)$ for every $Y$ satisfying~\eqref{eq:validrangeR}.
  In particular, we will obtain $\perm(A)= \myperm(A_{T_{\mathit{root}}},\chi(T_{\mathit{root}}))$.
  We will show that for each $t$ we can compute $\myperm(A_{T_t},Y)$ for all $Y$ in $\widetilde{O}((k_t+1)4^{\omega})$, where $k_t$ is the number of children of $t$.
  Note that $\sum_t k_t$ is the number of nodes of tree $T$.
  As the tree has $O(n)$ nodes, the total cost is then $\widetilde{O}(n\,4^{\omega})$, as wanted.

  The base case is when $t$ is a leaf of $T$, so that $A_{T_t} = A_{t}$ and $\chi(T_t) = \chi(t)$.
  Let $M_0$ be the submatrix of $M$ with rows $A_{t}$ and columns $\chi(t)$.
  Then all we have to do is to obtain the permanent of some submatrices of $M_0$.
  Observe that $\card{A_{t}}\leq \card{\chi(t)}\leq \omega$, as otherwise there is no $Y$ satisfying~\eqref{eq:validrangeR}.
  Thus, we can do this in $\widetilde{O}(2^{2\omega})$ using Lemma~\ref{thm:subperms}.

  Assume now that $t$ is an internal node of $T$ with $k_t$ children and let $Y$ that satisfies~\eqref{eq:validrangeR}.
  Then equation~\eqref{eq:recursepermcols} tells us how to find $\myperm(A_{T_t},Y)$.
  Lemma~\ref{thm:countpartitions} says that we can evaluate the formula in $\widetilde{O}(k_t\,2^{\omega})$, assuming we know the values of the terms in the recursion.
  Note that we already found $\myperm(A_{T_c},Y_c)$ for all children and that we can find $\myperm(A_{t},Y_t)$ for all possible $Y_t$ in $\widetilde{O}(2^{2\omega})$ in the same way as for the base case.
  Thus, it takes $\widetilde{O}(4^{\omega}+k_t\,2^{\omega})=\widetilde{O}((k_t+1)4^{\omega})$ to compute $\myperm(A_{T_t},Y)$ for all $Y$.
\end{proof}
\begin{rem}
  The factor $\widetilde{O}(4^\omega)$ in the proof came from the base case.
  This bound can be improved, but we omit this as for the approach of Section~\ref{s:bipartite}, based on the bipartite graph, the bound will be $\widetilde{O}(2^\omega)$.
  %It is easy to improve this bound, but we omit this as in Section~\ref{s:bipartite} the bound will be $\widetilde{O}(2^\omega)$.
\end{rem}

\subsection{Computing the determinant}\label{s:treedeterminant}

Given the similarity between permanent and determinant, it should be possible to find an analogous algorithm for the determinant.
We will derive such algorithm in this section.
Ironically, this algorithm is slower than the one for the permanent.
The reason is that the approach we follow does not take advantage of linear algebra: we loop over all permutations (carefully) and then compute its sign.
We remark that our algorithm does not use divisions and thus can be applied in any commutative ring.
The ideas from this section will be used in Section~\ref{s:discriminant} to derive a decomposition algorithm for the mixed discriminant.

\begin{exmp}
  Consider again the matrix $M$ of Figure~\ref{fig:triangles}, and observe that a similar expansion holds for the determinant:
  \begin{align*}
    \det(M) &= 
    \det\left(
\begin{bmatrix} M_{a_1,x_1} & M_{a_1,x_3}\\ M_{a_2,x_1} & M_{a_2,x_3} \end{bmatrix}
    \right)
    \det\left(
\begin{bmatrix} M_{a_3,x_4} \end{bmatrix}
    \right)
    \det\left(
\begin{bmatrix} M_{a_4,x_2} & M_{a_4,x_5}\\ M_{a_5,x_2} & M_{a_5,x_5} \end{bmatrix}
    \right)\\
    &-\det\left(
\begin{bmatrix} M_{a_1,x_1} & M_{a_1,x_4}\\ M_{a_2,x_1} & M_{a_2,x_4} \end{bmatrix}
    \right)
    \det\left(
\begin{bmatrix} M_{a_3,x_3} \end{bmatrix}
    \right)
    \det\left(
\begin{bmatrix} M_{a_4,x_2} & M_{a_4,x_5}\\ M_{a_5,x_2} & M_{a_5,x_5} \end{bmatrix}
    \right)\\
    &+\det\left(
\begin{bmatrix} M_{a_1,x_1} & M_{a_1,x_4}\\ M_{a_2,x_1} & M_{a_2,x_4} \end{bmatrix}
    \right)
    \det\left(
\begin{bmatrix} M_{a_3,x_2} \end{bmatrix}
    \right)
    \det\left(
\begin{bmatrix} M_{a_4,x_3} & M_{a_4,x_5}\\ M_{a_5,x_3} & M_{a_5,x_5} \end{bmatrix}
    \right).
  \end{align*}
As suggested in the above formula, the recursion used to compute the permanent can also be used to compute the determinant, by appropriately selecting the signs.
\end{exmp}

We recall now the definition of the parity function, and we extend it to ordered partitions.

\begin{defn}\label{defn:parityfunction}
Let $D$, $Y$ be ordered sets of the same size.
For a bijection $\pi:D\to Y$ we define its \emph{sign} or \emph{parity} as $\sgn(\pi):=(-1)^{N(\pi)}$, where $N(\pi)$ is its number of inversions:
\begin{align*}
  N(\pi) := \card{\{(a,a')\in D^2:\, a<a',\; \pi(a)>\pi(a')\}}.
\end{align*}
Let $\mathcal{Y}=(Y_1,\ldots,Y_k)$ be an \emph{ordered partition} of $Y$, i.e., $Y = Y_1\sqcup\cdots\sqcup Y_k$.
We define its \emph{sign} to be $\sgn(\mathcal{Y}):=(-1)^{N(\mathcal{Y})}$, where $N(\mathcal{Y})$ is:
\begin{align*}
  N(\mathcal{Y}) := \card{\{(y_i,y_j):\,y_i\in Y_i,\;y_j\in Y_j,\; i<j,\; y_i>y_j\}}.
\end{align*}
Equivalently, we can associate to $\mathcal{Y}$ a permutation $\pi^\mathcal{Y}: \{1,2,\ldots,\card{Y}\}\to Y$ that consists of blocks: we put first $Y_1$ (sorted), then $Y_2$ (sorted), and so on.
Then $\sgn(\mathcal{Y})=\sgn(\pi^{\mathcal{Y}})$.
%We define its \emph{sign} to be $\sgn(\mathcal{Y}):=\sgn(\pi^{\mathcal{Y}})$.
\end{defn}

From the definition above it is clear that we can obtain the sign of a permutation in $O(n^2)$ by counting the number of inversions.
However, it is well known that we can find it in $O(n)$ by counting its cycles.

For a matrix $M$, there is a natural order for its row set $A$ and column set $X$, namely from top to bottom and from left to right. 
We recall the definition of the determinant:
\begin{align*}
  \det(M) = \sum_\pi \sgn(\pi) \prod_{a \in A} M_{a,\pi(a)}
\end{align*}
where the sum is over all bijections $\pi:A\to X$.
Similarly, for a fixed matrix $M$ and for some $D\subset A$ and $Y\subset X$ we define the partial determinants:
\begin{align}\label{eq:defnpartialdet}
  \mydet(D,Y) := \sum_\pi \sgn(\pi)\prod_{a \in D} M_{a,\pi(a)} 
\end{align}
where the sum is over all bijections $\pi:D\to Y$.
Note that $\det(M)=\mydet(A,X)$.

We now provide a recursion formula similar to the one in Lemma~\ref{thm:recursepermcols}.
We need one lemma before.

\begin{lem}\label{thm:recursionsign}
  Let $D$, $Y$ be ordered sets, and let $\pi: D\to Y$ be a bijection, which we view as a subset of $D\times Y$.
  Let $\mathcal{D}=(D_1,\ldots,D_k)$ and $\mathcal{Y}=(Y_1,\ldots,Y_k)$ be partitions of $D$ and $Y$.
  Let $\pi= \pi_1\sqcup \cdots\sqcup \pi_k$ be a decomposition with $\pi_j\subset D_j\times Y_j$.
  %$\pi_j: D_j\to Y_j$.
  Then 
  \begin{align*}
    \sgn(\pi) = \sgn(\mathcal{D})\sgn(\mathcal{Y}) \prod_{j=1}^k\sgn(\pi_j).
  \end{align*}
\end{lem}
\begin{proof}
  It follows from the multiplicativity of the sign function.
\end{proof}

\begin{lem}\label{thm:recursedetcols}
  Under the same conditions of Lemma~\ref{thm:recursepermcols}, then
  \begin{align}\label{eq:recursedetcols}
    \mydet(A_{T_t}, Y) &= \sgn(\mathcal{D}) \sum_{\mathcal{Y}} \sgn(\mathcal{Y}) \mydet(A_{t},Y_t) \prod_{j=1}^k \mydet(A_{T_{c_j}}, Y_{c_j} ) 
  \end{align}
  where $\mydet(\cdot,\cdot)$ is as in \eqref{eq:defnpartialdet} the sum is over all $\mathcal{Y}=(Y_{s},Y_{c_1},\ldots,Y_{c_k})$ satisfying~\eqref{eq:partitionrangepi}, $\mathcal{D} = (A_{t},A_{T_{c_1}},\ldots,A_{T_{c_k}})$ and $\sgn(\cdot)$ is as in Definition~\ref{defn:parityfunction}.
\end{lem}
\begin{proof}
  The proof is basically the same as the one of Lemma~\ref{thm:recursepermcols}. 
  The only difference is that we have the additional factor $\sgn(\pi)$, but it factors because of Lemma~\ref{thm:recursionsign}.
\end{proof}

Despite the resemblance between equations~\eqref{eq:recursepermcols} and~\eqref{eq:recursedetcols}, the latter is not a subset convolution because of the sign factors.
Therefore, we cannot use the algorithm from~\cite{Bjoerklund2007} in this case.
We now show to the complexity analysis.

\begin{lem}\label{thm:evaluaterecursiondet}
  Given the values of the partial determinants $\mydet(A_t,Y_t)$ and $\mydet(A_{T_{c_j}},Y_{c_j})$, we can evaluate equation~\eqref{eq:recursedetcols} for all $Y$ satisfying~\eqref{eq:validrangeR} in $\widetilde{O}(k(n+ 3^{\omega}))$.
  %We can evaluate equation~\eqref{eq:recursedetcols} for all $Y$ satisfying~\eqref{eq:validrangeR} in $\widetilde{O}(n+ k\,2^{\omega})$, assuming we know the value of the partial determinants.
\end{lem}
\begin{proof}
  We will first express equation~\eqref{eq:recursedetcols} in a similar format as formula~\eqref{eq:recurseconv} of Lemma~\ref{thm:countpartitions}.
  To simplify the notation, let $\mathcal{Y} =: (Y_0,Y_1,\ldots, Y_{k})$.
  For each $j$ let 
  $Y_0^{j} = Y_0\cup Y_{1}\cup\cdots\cup Y_{j}$, and observe that
  $\sgn(\mathcal{Y}) = \prod_j \sgn(Y_0^{j-1},Y_j)$.
  Then equation~\eqref{eq:recursedetcols} can be rewritten as:
  \begin{align*}
    D(\bar{Y}) &= \sgn(\mathcal{D})\sum_{\bar{Y}_0\sqcup \cdots\sqcup\bar{Y}_k=\bar{Y}}\prod_{j=0}^k S_j(\bar{Y}_0^{j-1},\bar{Y}_j)\,D_{j}(\bar{Y}_{j} ) 
  \end{align*}
  where $\bar{Y}_0,\ldots,\bar{Y}_k\subset \chi(t)$ and
  \begin{gather*}
    \Delta:=\chi(T_{t})\setminus\chi(t) \qquad \Delta_0:=\emptyset \qquad \Delta_{j}:=\chi(T_{c_j})\setminus\chi(t)\\
    D(\bar{Y}) := \mydet(A_{T_t},\bar{Y}\cup\Delta),\quad 
    D_0(\bar{Y}_0) := \mydet(A_t,\bar{Y}_0\cup \Delta_0),\quad
    D_{j}(\bar{Y}_{j}) := \mydet(A_{T_{c_j}},\bar{Y}_j\cup \Delta_{j})\\
    \bar{Y}_0^j = \bar{Y}_0\cup\cdots\cup\bar{Y}_j \qquad \Delta_0^j:=\Delta_0\cup\cdots\cup \Delta_j\\
    S_j(\bar{Y}_0^{j-1},\bar{Y}_j):= \sgn(\bar{Y}_0^{j-1}\cup \Delta_0^{j-1},\, \bar{Y}_{j}\cup \Delta_{j})
  \end{gather*}

  For each $0\leq l \leq k$, and for each $\bar{Y}_0^l\subset \chi(t)$, let
  \begin{align*}
    D_0^l(\bar{Y}_0^l) = \sgn(\mathcal{D})\sum_{\bar{Y}_0\sqcup\cdots\sqcup \bar{Y}_l= \bar{Y}_0^l }\; \prod_{j=0}^l S_j(\bar{Y}_0^{j-1},\bar{Y}_j)\,D_{j}(\bar{Y}_{j} ) 
  \end{align*}
  Note that $D_0^k(\bar{Y}) = D(\bar{Y})$, and thus it is enough to compute $D_0^l$ for all $l$.
  We can do this recursively, observing that $D_0^0(\bar{Y}) = \sgn(\mathcal{D})D_0(\bar{Y})$ and
  \begin{align}\label{eq:recursiondet}
    D_0^{l+1}(\bar{Y}_0^{l+1}) &= \sum_{\bar{Y}_0^{l}\sqcup \bar{Y}_{l+1} = \bar{Y}_0^{l+1}} S_{l+1}(\bar{Y}_0^l,\bar{Y}_{l+1})\, D_0^{l}(\bar{Y}_0^{l})\,D_{l+1}(\bar{Y}_{l+1}) 
  \end{align}

  We reduced the problem to evaluating the above formula, and we will show that for each $l$ we can do this in $\widetilde{O}(n+3^\omega)$.
  Assume for now that the signs $S_{l+1}(\bar{Y}_0^l,\bar{Y}_{l+1})$ are known.
  Then for each $\bar{Y}_0^{l+1}$ of cardinality $i$, we can evaluate~\eqref{eq:recursiondet} in $O(2^{i})$.
  Thus, for all $\bar{Y}_0^{l+1}$ we require
    $O(\sum_{i} {w \choose i} 2^i) = O(3^{w})$,
  where $w=\card{\chi(t)}$.
  We will see that after a precomputation that takes $\widetilde{O}(n)$, we can obtain $S_{l+1}(\bar{Y}_0^l,\bar{Y}_{l+1})$ in $\widetilde{O}(1)$, which will complete the proof.
  
  Observe that
  \begin{align*}
    S_{l+1}(\bar{Y}_0^l,\bar{Y}_{l+1}) = 
    \sgn(\bar{Y}_0^l,\bar{Y}_{l+1})\,
    \sgn(\bar{Y}_0^l,\Delta_{l+1})\,
    \sgn(\Delta_0^l,\bar{Y}_{l+1})\,
    \sgn(\Delta_0^l,\Delta_{l+1}).
  \end{align*}
  Note that the last factor does not depend on the partition and it can be precomputed in $O(n)$.
  Also note that the first factor can be computed in $O(\omega)=\widetilde{O}(1)$, so we can ignore it.
  We are left with the second and third factor.

  For each $x\in \chi(t)$, let
  \begin{align*}
    N_x^{\Delta_{l+1}} &= \card{\{y\in \Delta_{l+1}:\, x>y\}}.
  \end{align*}
  We can precompute $N_x^{\Delta_{l+1}}$ for all $x$ in $O(\omega n)=\widetilde{O}(n)$.
  Note that $\sgn(\bar{Y}_0^l,\Delta_{l+1})= (-1)^N$ where $N = \sum_{x\in \bar{Y}_0^l} N_x^{\Delta_{l+1}}$.
  Thus, after the precomputation, we can obtain this factor in $O(\omega)=\widetilde{O}(1)$. 
  A similar procedure can be done for $\sgn(\Delta_0^l,\bar{Y}_{l+1})$.
\end{proof}

\begin{thm}\label{thm:treewidthdet}
  Let $M$ be a matrix with associated column graph $G^X$. 
  Let $(T,\chi)$ be a tree decomposition of $G^X$ of width $\omega$.
  Then we can compute $\det(M)$ in $\widetilde{O}(n^2+n\,4^{\omega})$.
\end{thm}
\begin{proof}
  There are two changes with respect to the proof of Theorem~\ref{thm:treewidthpermanent}. 
  First, in the base case we need to compute the determinant of all submatrices of $M_0$.
  Using an expansion by minors as in the proof of Lemma~\ref{thm:subperms}, we can do this in $\widetilde{O}(4^\omega)$, i.e., the same as for the permanent.
  Second, for the recursion formula we use Lemma~\ref{thm:evaluaterecursiondet}.
  This increases the time per node from $\widetilde{O}(k_t\,2^{\omega})$ to $\widetilde{O}(k_t(n+\,3^{\omega}))$.
  Therefore, the overall cost is $\widetilde{O}(n^2 + n\,4^{\omega})$.
\end{proof}

We conclude this section by presenting an open question.
Given the resemblance in the definition of the permanent and the determinant it is not surprising that they can be computed using very similar tree decomposition methods.
The \emph{immanant} of a matrix is another closely related notion:
\begin{align*}
  \mathrm{Imm}_\lambda(M) := \sum_\pi \chi_\lambda(\pi) \prod_{a \in A} M_{a,\pi(a)}
\end{align*}
where the sum is over all bijections $\pi:A\to X$ and $\chi_\lambda$ is an irreducible character of the symmetric group.
The immanant reduces to the permanent when $\chi_\lambda$ is the trivial character, and it reduces to the determinant when $\chi_\lambda$ is the sign character.
The computational complexity of immanants has been analyzed in e.g.,~\cite{Buergisser2000}.
A natural question that arises is whether a tree decomposition method can be used to compute them.
We remark that the recursion in~\eqref{eq:recursedetcols} does not hold for the immanant as $\chi_\lambda$ is not necessarily multiplicative.

\begin{ques}
  Given a matrix $M$ of bounded treewidth, can we compute $\mathrm{Imm}_\lambda(M)$ in polynomial time?
  In particular, can this be done if either the height or the width of the Young diagram is bounded?
\end{ques}

%%% Local Variables: 
%%% mode: latex
%%% TeX-master: "main"
%%% End: 

% Corrections Pablo -- July 7

\section{Bipartite decompositions}\label{s:bipartite}

In the previous section we showed a decomposition method based on the column graph $G^X$.
We showed that we can compute the permanent in $\widetilde{O}(n\,4^{\omega_X})$, where $\omega_X$ is the treewidth of $G^X$.
In this section we will extend this decomposition method to work with the bipartite graph $G$ (see Definition~\ref{defn:graphbipartite}).
We will show that we can compute the permanent in $\widetilde{O}(n\,2^{\omega})$, where $\omega$ is the treewidth of $G$.
A Matlab implementation of our algorithm is available in \url{www.mit.edu/~diegcif}.

Let $G$ be the bipartite graph of $M$.
As in the previous sections, we index the rows with a set $A$ and the columns with $X$.
%We consider a tree decomposition of $G$, and we will think of it in the following form.
We now rephrase the definition of a tree decomposition of $G$.
A \emph{bipartite decomposition} of $G$ is a tuple $(T,\alpha,\chi)$, where $T$ is a rooted tree and $\alpha: T \to \{0,1\}^A$, $\chi:T\to \{0,1\}^X$ assign some $\alpha(t)\subset A$ and $\chi(t)\subset X$ to each node $t$ of $T$, that satisfies the following conditions.
\begin{enumerate}
  \item[i-1.] The union of $\{\alpha(t)\}_{t\in T}$ is the whole row set $A$.
  \item[i-2.] The union of $\{\chi(t)\}_{t\in T}$ is the whole column set $X$.
  \item[ii.] For every edge $(a, x)$ of $G$ there exists a node $t$ of $T$ with $a\in \alpha(t),x\in \chi(t)$.
  \item[iii-1.] For every $a\in A$ the set $\{t: a\in \alpha(t)\}$ forms a subtree of $T$. 
  \item[iii-2.] For every $x\in X$ the set $\{t: x\in \chi(t)\}$ forms a subtree of $T$. 
\end{enumerate}
The width $\omega$ of the decomposition is the largest of $\card{\alpha(t)}+\card{\chi(t)}$ among all nodes $t$.
Note that the above literals are consistent with the ones in Definition~\ref{defn:treedecomposition}.
%The width of the decomposition is the largest of $\max\{\card{\alpha(t)},\card{\chi(t)}\}$ among all $t\in T$.

As before, we now present an example to illustrate the use of the bipartite graph for computing the permanent.

\begin{exmp}
  \begin{figure}[htb]
  \centering
  \includegraphics[scale=0.42]{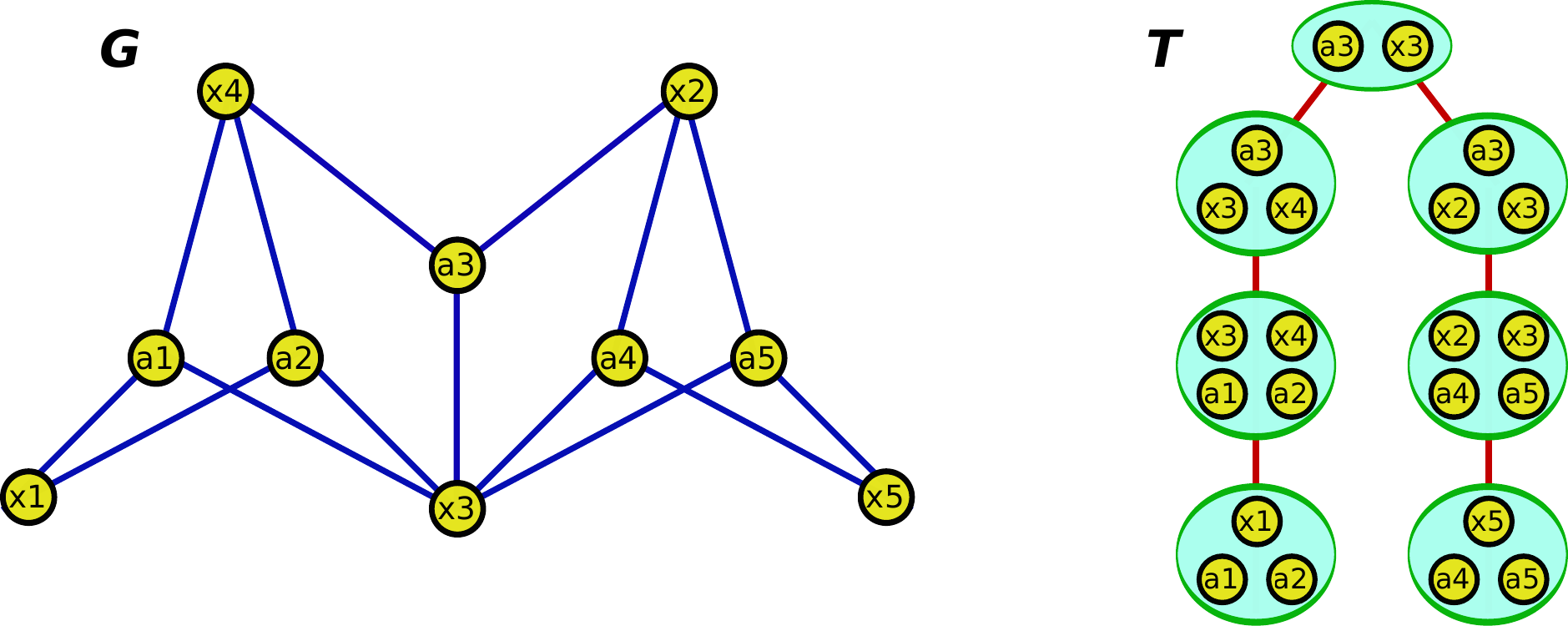}
  \vspace{-5 pt}
  \caption{Bipartite graph $G$ of the matrix of Figure~\ref{fig:triangles} and a tree decomposition $T$.}
  \label{fig:trianglesbipart}
  \end{figure}
  Consider again the matrix $M$ of Figure~\ref{fig:triangles}.
  Note that $\perm(M)$ can also be expressed in the following form:
  {\small
  \begin{align*}
    \perm(M) &= \myperm(\{a_1,a_2\},\{x_1,x_4\})\,\myperm(\{a_3,a_4,a_5\},\{x_2,x_3,x_5\})\\
    &+ \myperm(\{a_1,a_2,a_3\},\{x_1,x_3,x_4\})\,\myperm(\{a_4,a_5\},\{x_2,x_5\})\\
    &- M_{a_3,x_3}\,\myperm(\{a_1,a_2\},\{x_1,x_4\})\,\myperm(\{a_4,a_5\},\{x_2,x_5\})\\
    \myperm(\{a_1,a_2,a_3\}&,\{x_1,x_3,x_4\}) = M_{a_3,x_3}\,\myperm(\{a_1,a_2\},\{x_1,x_4\}) + M_{a_3,x_4}\,\myperm(\{a_1,a_2\},\{x_1,x_3\})\\
    \myperm(\{a_3,a_4,a_5\}&,\{x_2,x_3,x_5\}) = M_{a_3,x_3}\,\myperm(\{a_4,a_5\},\{x_2,x_5\}) + M_{a_3,x_2}\,\myperm(\{a_4,a_5\},\{x_3,x_5\})
  \end{align*}
}
  To evaluate the above formula we need to compute four $2\times 2$ permanents, and we need in total $16$ multiplications.
  It turns out that this formula arises by considering the tree decomposition of the bipartite graph shown in Figure~\ref{fig:trianglesbipart}.
\end{exmp}

\subsection{Bipartite decomposition algorithm}

Algorithm~\ref{alg:bipartiteperm} presents our dynamic programming method to compute $\perm(M)$ using a bipartite decomposition.
As for Algorithm~\ref{alg:colsperm}, for each node $t$ we compute a table $P_t$, following a topological ordering of the tree.
The permanent of $M$ is in the table corresponding to the root.
There are two main routines: \Call{SubPerm}{} computes the permanents of all submatrices, and \Call{EvalRecursion}{} evaluates a recursion formula, which is slightly more complex than the one of Algorithm~\ref{alg:colsperm}.

\begin{algorithm}
  \caption{Permanent with bipartite decomposition}
  \label{alg:bipartiteperm}
  \begin{algorithmic}[1]
    \Require{Matrix $M$ and tree decomposition $(T,\alpha,\chi)$ of bipartite graph $G(M)$}
    \Ensure{Permanent of $M$}
    \Procedure{BipartPerm}{$M,T,\alpha,\chi$}
    \State $\mathit{order} := $ topological ordering of $T$ starting from its leaves
    \For {$t$ in $\mathit{order}$}
    \State $Q_t:=$ \Call{SubPerms}{$t,M$}
    \If{$t$ is a leaf}
    \State $P_t := Q_t$
    \Else
    \State ${c_1},\ldots,{c_k}:=$ children of $t$
    \State $P_t := $\Call{EvalRecursion}{$t,Q_t,P_{c_1},\ldots,P_{c_k}$}
    \EndIf
    \EndFor
    \State \Return $P_{\mathit{root}}(\alpha(\mathit{root}),\chi(\mathit{root}))$
    \EndProcedure
    \Procedure{SubPerms}{$t,M$}
    \State $Q_t(\bar{D},\bar{Y}) := \myperm(\bar{D},\bar{Y})$ \; \algorithmicforall\ $ \bar{D}\subset \alpha(t), \bar{Y}\subset\chi(t)$
    \EndProcedure
    \Procedure{EvalRecursion}{$t,Q_t,P_{c_1},\ldots,P_{c_k}$}
    \For {$c_j$ child of $t$}
    \State $\Delta^{\alpha}_j:= \alpha(c_j)\setminus \alpha(t)$,\quad $\Lambda^{\alpha}_j:= \alpha(c_j)\cap \alpha(t)$
    \State $\Delta^{\chi}_j:= \chi(c_j)\setminus \chi(t)$,\quad $\Lambda^{\chi}_j:= \chi(c_j)\cap \chi(t)$
    \State $Q_{tc_j}(\bar{D},\bar{Y}) := (-1)^{\card{\bar{D}}}Q_{t}(\bar{D}, \bar{Y})$ \; \algorithmicforall\ $ \bar{D}\subset \Lambda^{\alpha}_j,\bar{Y}\subset\Lambda^{\chi}_j$
    \State $Q_{cc_j}(\bar{D},\bar{Y}) := P_{c_j}(\bar{D}\cup\Delta^\alpha_j, \bar{Y}\cup \Delta^\chi_j)$ \; \algorithmicforall\ $ \bar{D}\subset \Lambda^{\alpha}_j,\bar{Y}\subset\Lambda^{\chi}_j$
    \EndFor
    \State $P_t:=$ \Call{SubsetConvolution}{$Q_t,Q_{tc_1},Q_{cc_1},\ldots,Q_{tc_k},Q_{cc_k}$}
    \EndProcedure
    \Procedure{SubsetConvolution}{$P_0,P_1,\ldots,P_{2k}$}
    \State $P(\bar{D},\bar{Y}) := \displaystyle\sum\limits_
    {\substack{\bar{D}_0\sqcup\cdots\sqcup \bar{D}_{2k}= \bar{D}\\
    \bar{Y}_0\sqcup\cdots\sqcup \bar{Y}_{2k}= \bar{Y}}}
    P_0(\bar{D}_0,\bar{Y}_0) \,P_1(\bar{D}_1,\bar{Y}_1)\cdots P_{2k}(\bar{D}_{2k},\bar{Y}_{2k})$  
%\;\; \algorithmicforall\ $ \bar{Y}$
    \EndProcedure
  \end{algorithmic}
\end{algorithm}

As before, the values $P_t(\bar{D},\bar{Y})$ computed correspond to a partial permanent of the matrix.
Consider the collection
  \begin{align*}
    S = \{(D,Y):\alpha(T_t) \setminus  \alpha(t)\subset D \subset \alpha(T_t),\,
    \chi(T_t) \setminus  \chi(t)\subset Y \subset \chi(T_t)\}.
  \end{align*}
  Observe that $(D,Y)\in S$ is completely determined by $(D\cap\alpha(t),Y\cap\chi(t))$.
  Therefore, if we let $\bar{D}=D\cap \alpha(t),$ $\bar{Y}:=Y\cap \chi(t)$, there is a one to one correspondence between $S$, and the collection $\bar{S} := \{(\bar{D},\bar{Y}):\bar{D}\subset \alpha(t),\bar{Y}\subset \chi(t)\}$.
The partial permanents that we are interested in are
\begin{align*}
  P_t(\bar{D},\bar{Y}) := \myperm(D,Y) = \myperm(\bar{D}\cup (\alpha(T_t)\setminus\alpha(t)),\,\bar{Y}\cup (\chi(T_t)\setminus \chi(t))).
\end{align*}
For the root node $P_{\mathit{root}}(\alpha(\mathit{root}),\chi(\mathit{root}))=\myperm(A,X)=\perm(M)$.

\subsection{Recursion formula}

The recursion formula that method $\Call{EvalRecursion}{}$ of Algorithm~\ref{alg:bipartiteperm} evaluates is given in the following lemma.\\
%This is the analogue of Lemma~\ref{thm:recursepermcols}.

\begin{lem}\label{thm:recursepermbipart}
  Let $M$ be a matrix with associated bipartite graph $G$. 
  Let $(T,\alpha,\chi)$ be a bipartite decomposition of $G$.
  Let $t$ be an internal node of $T$, let $T_t\subset T$ denote the subtree rooted in $t$, and  let $D,Y$ be such that
  \begin{align}
  \alpha(T_t)\setminus \alpha(t)\subset D\subset \alpha(T_t),\;\;\;\;\;\;\;
  \chi(T_t)\setminus \chi(t)\subset Y\subset \chi(T_t),\;\;\;\;\;\;\;
  \card{D} &= \card{Y}
    \label{eq:validpairDR}
  \end{align}
  Let $t_{c_1},\ldots,t_{c_k}$ be the children of $t$. Then 
  \begin{align}\label{eq:recursepermDRfull}
    \myperm(D,Y) &= \sum_{\mathcal{D},\mathcal{Y}} \myperm(D_t,Y_t) \prod_{j=1}^k (-1)^{\card{D_{tc_j}}}  \myperm(D_{tc_j},Y_{tc_j})\myperm(D_{cc_j},Y_{cc_j})
  \end{align}
  where $\myperm(\cdot,\cdot)$ is as in \eqref{eq:defnpartialperm} and the sum is over all $\mathcal{D}=(D_t,D_{tc_1},\ldots)$, $\mathcal{Y}=(Y_t,Y_{tc_1},\ldots)$ satisfying:
  \begin{align*}
    D &=  D_t \sqcup (D_{tc_1}\sqcup D_{cc_1} \sqcup \cdots \sqcup D_{tc_k}\sqcup D_{cc_k}) \\
    Y &=  Y_t \sqcup (Y_{tc_1}\sqcup Y_{cc_1} \sqcup \cdots \sqcup Y_{tc_k}\sqcup Y_{cc_k}) \\
    \alpha(T_{c_j})\setminus  \alpha(t) &\subset D_{cc_j} \subset \alpha(T_{c_j}) \;\;\;\;\;\;
   D_{t} \subset \alpha(t)  \;\;\;\;\;\;
   D_{tc_j} \subset \alpha(t)\cap\alpha(t_{c_j})\\
   \chi(T_{c_j})\setminus  \chi(t) &\subset Y_{cc_j} \subset \chi(T_{c_j}) \;\;\;\;\;\;
   Y_{t} \subset \chi(t)  \;\;\;\;\;\;
   Y_{tc_j} \subset \chi(t)\cap\chi(t_{c_j}).
  \end{align*}
\end{lem}

To prove this lemma we need some additional notation.
We view a bijection $\pi:D\to Y$ as a subset of $D\times Y$, by identifying it with the set $\{(a,\pi(a)): a\in D\}$.
For a given node $t$ and for some $D,Y$ satisfying~\eqref{eq:validpairDR}, we denote
\begin{align*}
  \myperm^*(D,Y) := \sum_{\pi} \prod_{a \in D} M_{a,\pi(a)} 
\end{align*}
where the sum is over all bijections $\pi:D\to Y$ such that
  \begin{align}\label{eq:overcountpi}
    \pi\cap (\alpha(t)\times\chi(t)) = \emptyset.
  \end{align}
We now show a different recursion formula, which is closer to the one in Lemma~\ref{thm:recursepermcols}.

\begin{lem}\label{thm:recursepermbipart2}
  Following the same notation as above, the following equation holds:
  \begin{align}\label{eq:recursivepermDR}
    \myperm(D,Y) &= \sum_{Y_t,Y_{c_j},D_t,D_{c_j}} \myperm(D_t,Y_t) \prod_{j=1}^k \myperm^*(D_{c_j},Y_{c_j}) 
  \end{align}
  where the sum is over all $Y_t,Y_{c_j},D_t,D_{c_j}$ such that
  \begin{subequations}\label{eq:partitiondomainrangepi}
  \begin{align}
   D &=  D_t \sqcup (D_{c_1} \sqcup \cdots \sqcup D_{c_k}) \\
   Y &=  Y_t \sqcup (Y_{c_1} \sqcup \cdots \sqcup Y_{c_k}) \\
   \alpha(T_{c_j})\setminus  \alpha(t) &\subset D_{c_j} \subset \alpha(T_{c_j}) \;\;\;\;\;\;
   D_t \subset \alpha(t)\\
   \chi(T_{c_j})\setminus  \chi(t) &\subset Y_{c_j} \subset \chi(T_{c_j}) \;\;\;\;\;\;
   Y_t \subset \chi(t).
  \end{align}
  \end{subequations}
\end{lem}

\begin{proof}
  Let $\pi:D\to Y$ be a matching, which we view as a subset of $D\times Y$.
  Note that
  \begin{align*}
    D &= (D\cap \alpha(t))\sqcup (D\cap \alpha(T_{c_1})\setminus \alpha(t)) \sqcup \cdots \sqcup (D\cap \alpha(T_{c_k})\setminus \alpha(t))\\
    Y &= (Y\cap \chi(t))\sqcup (Y\cap \chi(T_{c_1})\setminus \chi(t)) \sqcup \cdots \sqcup (Y\cap \chi(T_{c_k})\setminus \chi(t))
  \end{align*}
  Let's decompose $\pi$ in a similar way as above.
  Let $\pi_t$ be the submatching of $\pi$ with domain contained in $D\cap \alpha(t)$ and range contained in $Y\cap\chi(t)$.
  Equivalently, $\pi_t = \pi\cap (D\cap \alpha(t))\times (Y\cap\chi(t))$.
  Observe that if some $a\in D\cap \alpha(t)$ is not in the domain of $\pi_t$, then $a \in \alpha(T_c),\pi(a)\in \chi(T_c)$ for some child $c$.
  The reason is that by definition of tree decomposition, there must be some node $t_p$ with $\pi(a)\in \chi(t_p),a\in \alpha(t_p)$.
  However, the assumption on $a$ says that $\pi(a)\in \chi(T_t)\setminus \chi(t)$ and thus  we must have $t_p \in T_{c}$ for some $c$.
  Similarly, if some $x \in Y\cap \chi(t)$ is not in the range of $\pi_t$, then $x\in \chi(T_c),\pi^{-1}(x)\in \alpha(T_c)$ for some child $c$.
  In the same way, we have the following
  \begin{align*}
  \pi(D\cap\alpha(T_c)\setminus \alpha(t)) &\subset D\cap\chi(T_c) \\ 
  \pi^{-1}(Y\cap\chi(T_c)\setminus \chi(t))&\subset Y\cap\alpha(T_c)
  \end{align*}

  The above paragraph implies that we can decompose 
  \begin{align}\label{eq:partitionfullpi}
  \pi = \pi_t \sqcup \pi_{c_1}\sqcup \cdots\sqcup \pi_{c_k}
  \end{align}
  in such a way that $\pi_t\subset (D\cap\alpha(t))\times(Y\cap\chi(t))$ and for each child $c$ we have that $\pi_{c} \subset (D\cap\alpha(T_c))\times(Y\cap\chi(T_c))$ and $\pi_c$ satisfies~\eqref{eq:overcountpi}.
  Moreover, it is easy to see that such decomposition is unique.
%  To see this, note first that $\pi_t$ has to be as defined above.
%  In addition, if $(a,x)\in \pi_c$ then either $a\in D\cap\alpha(T_c)\setminus \alpha(t)$ or $x\in Y\cap\chi(T_c)\setminus \chi(t)$, and thus child $c$ is determined, too.

  Let $Y_t \subset Y\cap \chi(t)$ be the range of $\pi_t$ and let $Y_c\subset Y\cap\chi(T_c)$ be the range of $\pi_c$.
  Analogously, define $D_t\subset D\cap\alpha(t)$ and $D_c\subset D\cap\alpha(T_c)$ as the domains of $\pi_t,\pi_c$. 
  Observe that $\chi(T_c)\setminus  \chi(t) \subset Y_c$, as if $x \in \chi(T_c)\setminus \chi(t)$ then $(\pi^{-1}(x),x)\in \pi_c$ by construction, and thus  $x\in Y_c$.
  Similarly, $\chi(T_c)\setminus  \chi(t) \subset Y_c$.
  Therefore, the equations~\eqref{eq:partitiondomainrangepi} are satisfied.

  Thus, for any matching $\pi\subset D\times Y$ there is a unique partition as in \eqref{eq:partitionfullpi} such that the ranges and domains $Y_t,Y_c,D_t,D_c$ satisfy \eqref{eq:partitiondomainrangepi}.
  On the other hand, assume that $Y_t,Y_c,D_t,D_c$  satisfy \eqref{eq:partitiondomainrangepi}, and we are given some matchings $\pi_t,\pi_c$ with these ranges and domains and such that \eqref{eq:overcountpi} holds.
  Then equations \eqref{eq:partitiondomainrangepi} tell us that we can merge them into a matching $\pi$ with domain $D$ and range $Y$.
  Condition \eqref{eq:overcountpi} ensures we are not overcounting, as it implies that decomposition~\eqref{eq:partitionfullpi} is unique.
  These remarks imply equation~\eqref{eq:recursivepermDR}.
\end{proof}
\begin{rem}
  Let $(T^X,\chi)$ be a tree decomposition of the column graph $G^X$, and let $(T,\alpha,\chi)$ be the corresponding decomposition of the bipartite graph $G$ given in Example~\ref{exmp:GbetterGX}.
  In such case, the above lemma reduces to Lemma~\ref{thm:recursepermcols}.
\end{rem}

We now derive the recursion formula in Lemma~\ref{thm:recursepermbipart}, which follows from the above lemma by using inclusion-exclusion.

\begin{proof}[Proof of Lemma~\ref{thm:recursepermbipart}]
  For a child $c$ of $t$, we will show that 
  \begin{align}\label{eq:count_inc_exc}
    \myperm^*(D_c,Y_c) &=  \sum_{
      \substack{
        D_{tc}\sqcup D_{cc} = D_c \\
        Y_{tc}\sqcup Y_{cc} = Y_c\\
        D_{tc}\subset \alpha(t), Y_{tc}\subset \chi(t)}
    }(-1)^{\card{D_{tc}}}  \myperm(D_{tc},Y_{tc})\myperm(D_{cc},Y_{cc}).
  \end{align}
  Combining equations~\eqref{eq:recursivepermDR} and~\eqref{eq:count_inc_exc} we obtain equation~\eqref{eq:recursepermDRfull}, concluding the proof.

  Given a matching $\pi_c:D_c\to Y_c$, let $I(\pi_c):=\pi_c\cap (\alpha(t)\times\chi(t))$ and let $I_\alpha(\pi_c)\subset \alpha(t)$, $I_\chi(\pi_c)\subset \chi(t)$ be the domain and range of $I(\pi_c)$.
  For some $D_{tc}\subset D_c\cap \alpha(t)$, $Y_{tc}\subset Y_c\cap \chi(t)$ with $\card{D_{tc}}=\card{Y_{tc}}$, let 
  \begin{align*}
    \myperm^*(D_c,Y_c;D_{tc},Y_{tc}) := 
    \sum_{\substack{\
      \pi_c:D_c\to Y_c\\ I_\alpha(\pi_c)=D_{tc}\\ I_\chi(\pi_c)=Y_{tc}}}
      \prod_{a}M_{a,\pi_c(a)}.
  \end{align*}
  Note that $\myperm^*(D_c,Y_c) = \myperm^*(D_c,Y_c;\emptyset,\emptyset)$.
  Observe now that given matchings $\pi_{tc}:D_{tc}\to Y_{tc}$ and $\pi_{cc}:D_c\setminus D_{tc}\to Y_c\setminus Y_{tc}$, we can merge them into a matching $\pi_c^*: D_c\to Y_c$ that satisfies $ I_\alpha(\pi_c^*)\supset D_{tc}$, $I_\chi(\pi_c^*)\supset Y_{tc}$.
  Therefore, we have the following equation
  \begin{align*}
    \myperm(D_{tc},Y_{tc})\,\myperm(D_{c}\setminus D_{tc},Y_{c}\setminus Y_{tc}) = 
    \sum_{\substack{D_{tc}^*\supset D_{tc}\\Y_{tc}^*\supset Y_{tc}}}
    \myperm^*(D_c,Y_c;D_{tc}^*,Y_{tc}^*).
  \end{align*}

  Based on the above formula, we can now find $\myperm^*(D_c,Y_c)$ using inclusion-exclusion (or M\"obius inversion):
  \begin{align*}
    \myperm^*(D_c,Y_c;\emptyset,\emptyset) &= \sum_{i} \sum_{
      \substack{
        D_{tc}\subset D_c\cap \alpha(t)\\
        Y_{tc}\subset Y_c\cap \chi(t)\\
        \card{D_{tc}}=\card{Y_{tc}}=i}
      }(-1)^i \myperm(D_{tc},Y_{tc})\,\myperm(D_c\setminus D_{tc},Y_c\setminus Y_{tc}).
  \end{align*}
  Rewriting the above equation leads to~\eqref{eq:count_inc_exc}, as wanted.
\end{proof}

\subsection{Complexity analysis}

We just derived the recursion formula~\eqref{eq:recursepermDRfull} which is used in Algorithm~\ref{alg:bipartiteperm}.
As in the proof of Lemma~\ref{thm:countpartitions}, this formula is a subset convolution and thus it can be evaluated efficiently using the algorithm from~\cite{Bjoerklund2007}.
%To evaluate the formula of equation~\eqref{eq:recursepermDRfull}, we use again the procedure of Lemma~\ref{thm:countpartitions}.
%This is done in the method \Call{EvalRecursion}{} of Algorithm~\ref{alg:bipartiteperm}.
The overall running time of Algorithm~\ref{alg:bipartiteperm} is stated now.

\begin{thm}\label{thm:bipartitepermanent}
  Let $M$ be a matrix with associated bipartite graph $G$. 
  Let $(T,\alpha,\chi)$ be a bipartite decomposition of $G$ of width $\omega$.
  Then we can compute $\perm(M)$ in $\widetilde{O}(n\, 2^{\omega})$.
\end{thm}
\begin{proof}
  The proof is very similar to the one of Theorem~\ref{thm:treewidthpermanent}.
  For each node $t\in T$, we will compute $\myperm(D,Y)$ for every pair $D,Y$ that satisfies equation~\eqref{eq:validpairDR}.
  We will show that for each $t$ we can compute $\myperm(D,Y)$ for all such $D,Y$ in $\widetilde{O}((k_t+1)2^{\omega})$, where $k_t$ is the number of children of $t$.
  Observe that for the root node $t_r$ we will compute $\perm(M)= \myperm(\alpha(T_r),\chi(T_r))$.
  This will conclude the proof.

  The base case is when $t$ is a leaf of $T$, so that $T_t = \{t\}$.
  Let $M_0$ be the submatrix of $M$ with rows $\alpha(t)$ and columns $\chi(t)$.
  We need to obtain the permanent of all submatrices of $M_0$.
  As $\card{\alpha(t)}+\card{\chi(t)}\leq \omega$, we can do this in $\widetilde{O}(2^{\omega})$ using Lemma~\ref{thm:subperms}.

  Assume now that $t$ is an internal node of $T$ with $k_t$ children and let $D,Y$ that satisfy~\eqref{eq:validpairDR}.
  Then equation~\eqref{eq:recursepermDRfull} tells us how to find $\myperm(D,Y)$.
  Similarly as in Lemma~\ref{thm:countpartitions}, we can evaluate this formula in $\widetilde{O}(k_t\,2^{\omega})$, assuming we know the values of the terms in the recursion.
  Note that we already found $\myperm(D_{cc},Y_{cc})$ for all children in the recursion.
  We can find $\myperm(D_t,Y_t)$ for all $D_t,Y_t$ in $\widetilde{O}(2^{\omega})$ in the same way as for the base case, and this includes the values $\myperm(D_{tc},Y_{tc})$.
  Then, it takes $\widetilde{O}(2^{\omega}+k_t\,2^{\omega})=\widetilde{O}((k_t+1)2^{\omega})$ to compute $\myperm(D,Y)$ for a all $D,Y$.
\end{proof}

Similarly as in Theorem~\ref{thm:treewidthdet}, we can find an analogous algorithm for the determinant.

\begin{thm}\label{thm:bipartitedet}
  Let $M$ be a matrix with associated bipartite graph $G$. 
  Let $(T,\alpha,\chi)$ be a bipartite decomposition of $G$ of width $\omega$.
  Then we can compute $\det(M)$ in $\widetilde{O}(n^2+n\,3^{\omega})$.
\end{thm}
\begin{proof}
  We just need to follow the steps of Section~\ref{s:treedeterminant}.
  For instance, the recursion is
  \begin{align*}%\label{eq:recursedetDRfull}
    \mydet(D,Y) &= \sum_{\mathcal{D},\mathcal{Y}}\sgn(\mathcal{D})\sgn(\mathcal{Y}) \mydet(D_t,Y_t) \prod_{j=1}^k (-1)^{\card{D_{tc_j}}}  \mydet(D_{tc_j},Y_{tc_j})\mydet(D_{cc_j},Y_{cc_j})
  \end{align*}
  where $\mathcal{Y} = (Y_t,Y_{tc},Y_{cc}),\mathcal{D} =(D_t,D_{tc},D_{cc})$.
  The complexity analysis is basically the same as in the proof of Theorem~\ref{thm:treewidthdet}.
  The base case can be done in $\widetilde{O}(2^\omega)$ using an expansion by minors.
  The recursion can be evaluated in $\widetilde{O}(k_t(n+3^\omega))$ in a similar way as in Lemma~\ref{thm:evaluaterecursiondet}.
\end{proof}

%%% Local Variables: 
%%% mode: latex
%%% TeX-master: "main"
%%% End: 

% Corrections Pablo -- June 1
% Subset convolution more explicit now
% Complexity determinant 3^n

\section{Mixed discriminant and higher dimensions}\label{s:discriminant}

The mixed discriminant of $n$ matrices is a common generalization of the permanent and the determinant.
As such, it is also hard to compute in the general case.
We show now that the techniques presented earlier generalize to compute mixed discriminants.
Even more, we show that this method extends to compute similar functions in higher dimensional tensors.

\subsection{Mixed discriminant}

Let $M$ be a list of $n$ matrices of size $n\times n$.
Equivalently, we can think of $M$ as a $n\times n\times n$ array.
We index the first coordinate with a set $A$, and the second and third coordinates with sets $X^1,X^2$.
The mixed discriminant of $M$ is given by
\begin{align*}
  \mdisc(M) := \sum_{\pi^1,\pi^2}\sgn(\pi^1)\sgn(\pi^2) \prod_{a\in A} M_{a,\pi^1(a),\pi^2(a)}
\end{align*}
where the sum is over all bijections $\pi^1: A\to X^1$ and $\pi^2: A\to X^2$, and $\sgn$ is the parity function.
For $a\in A$, let $M_a$ denote the $n\times n$ matrix obtained by fixing the first coordinate.
Observe that if $M_a = m$ for some matrix $m$ and for all $a\in A$, then $\mdisc(M) = n!\,\det(m)$.
In the case that $M_a$ is diagonal for all $a\in A$, then $\mdisc(M)= \perm(D)$ where $D$ is the matrix obtained by concatenating these diagonals.
Some of the properties of mixed discriminant are discussed in~\cite{Bapat1989}.

In the case of a $n\times n$ matrix, a bipartite graph was the natural structure to represent its sparsity.
Similarly, if we are given a sparse $n\times n\times n$ array $M$, a natural structure is a \emph{tripartite graph} $G$, as follows.
Let $G$ be the graph on $A\cup X^1\cup X^2$, where for each nonzero entry $M_{a,x_1,x_2}$ we put a triangle $\{a,x_1,x_2\}$.

%We consider now a tree decomposition of this tripartite graph.
We rephrase the definition of a tree decomposition of a tripartite graph $G$.
A \emph{tripartite decomposition} of $G$ is a tuple $(T,\alpha,\chi^1,\chi^2)$, where $T$ is a rooted tree, $\alpha: T\to \{0,1\}^A$, $\chi^1: T \to \{0,1\}^{X^1}$ and $\chi^2:T\to \{0,1\}^{X^2}$, that satisfies the following conditions.
\begin{enumerate}[ i.]
  \item The union of $\{\alpha(t)\}_{t\in T}$ (resp. $\chi^1$, $\chi^2$) is the whole $A$ (resp. $X^1$, $X^2$).
  \item For every triangle $(a, x_1,x_2)$ in $G$ there is a $t$ with $(a,x_1,x_2)\in (\alpha\times \chi^1\times\chi^2)(t)$.
  \item For every $a\in A$ (resp. $X^1$, $X^2$) the set $\{t: a\in \alpha(t)\}$ is a subtree of $T$. 
\end{enumerate}
The width of the decomposition is the largest of $\card{\alpha(t)}+\card{\chi^1(t)}+ \card{\chi^2(t)}$ among all nodes $t$.
Note that the above literals are consistent with the ones in Definition~\ref{defn:treedecomposition}.
In particular, observe that the second condition does not impose additional constraints due to Lemma~\ref{thm:treedecompclique}.

We proceed to extend the previous results to the mixed discriminant.
For some sets $D\subset A$, $Y^1\subset X^1$ and $Y^2\subset X^2$ we denote 
\begin{align}\label{eq:defnpartialdisc}
  \mydisc(D,Y^1,Y^2) := \sum_{\pi^1,\pi^2}\sgn(\pi^1)\sgn(\pi^2) \prod_{a \in D} M_{\pi^1(a)\pi^2(a)} 
\end{align}
where the sum is over all bijections $\pi^1:D\to Y^1$ and $\pi^2: D\to Y^2$.
This only makes sense if $\card{D}=\card{Y^1}=\card{Y^2}$, and otherwise we can define $\mydisc(D,Y^1,Y^2)=0$.

As for the case of the permanent, the dynamic program to compute $\mdisc(M)$ has two main steps: computing the mixed discriminant of all subarrays of $M$, and evaluating some recursion formula.
For the first step, it is easy to see that the approach from Lemma~\ref{thm:subperms} extends, as we show now.

\begin{lem}\label{thm:subdisc}
  Let $M_0$ be a $n_1\times n_2\times n_3$ array.
  Let $A_0,X_0^1,X_0^2$ be its set of coordinates, and let $S=\{(D,Y^1,Y^2)\subset A_0\times X_0^1\times X_0^2:\, \card{D}=\card{Y^1}=\card{Y^2}\}$.
  We can compute $\mydisc(D,Y^1,Y^2)$ for all triples in $S$ in $O(n_{max}^3\,2^{n_1+n_2+n_3})$, where $n_{max}=\max\{n_1,n_2,n_3\}$.
\end{lem}
\begin{proof}
  For $i=1,2,\ldots,\min\{n_1,n_2,n_3\}$ let
  \begin{align*}
    S_i := \{(D,Y^1,Y^2)\subset A_0\times X_0^1\times X_0^2:\, \card{D}=\card{Y^1}=\card{Y^2} = i\}.
  \end{align*}
  We can find $\mydisc(D,Y^1,Y^2)$ for all $(D,Y^1,Y^2)\in S_i$ using the values of $S_{i-1}$ as follows.
  Let $a_0$ be the first element in $D$, it is easy to see that
  \begin{align*}
    \mydisc(D,Y^1,Y^2) = \sum_{x_1\in Y^1,x_2\in Y^2}\epsilon(x_1,x_2) M_{a_0,x_1,x_2}\,\mydisc(D\setminus a_0, Y^1\setminus x_1, Y^2\setminus x_2)
  \end{align*}
  where $\epsilon(x_1,x_2)$ is either $+1$ or $-1$.
  To be concrete, if we identify $Y^1,Y^2$ with the set $\{1,\ldots,i\}$, then $\epsilon(j_1,j_2) = (-1)^{j_1+j_2}$.
  Thus, for each triple $(D,Y^1,Y^2)$ we just need to loop over $n_2n_3$ terms, and for each we need $O(n_{max})$ to find $D\setminus a_0, Y^1\setminus x_1, Y^2\setminus x_2$.
\end{proof}

The recursion formula we need to evaluate is given in the following lemma.

\begin{lem}\label{thm:recursedisc}
  Let $M$ be a list of $n$ matrices of size $n\times n$, with associated tripartite graph $G$. 
  Let $(T,\alpha,\chi^1,\chi^2)$ be a tripartite decomposition of $G$.
  Let $t$ be an internal node of $T$, let $T_t\subset T$ denote the subtree rooted in $t$, and  let $D,Y^1,Y^2$ be such that
  \begin{subequations} \label{eq:validtripleDYR}
  \begin{align}
  \alpha(T_t)\setminus \alpha(t)&\subset D\subset \alpha(T_t),\;\;\;\;\;\;\;\;\;\;
  \card{D} = \card{Y^1} = \card{Y^2}\\
  \chi^1(T_t)\setminus \chi^1(t)&\subset Y^1\subset \chi^1(T_t),\;\;\;
  \chi^2(T_t)\setminus \chi^2(t)\subset Y^2\subset \chi^2(T_t)%,\;\;\; 
  \end{align}
  \end{subequations}
  Let $c_1,\ldots,c_k$ be the children of $t$. Then 
  \begin{align}\label{eq:recursediscDYRfull}
    \begin{split}
    \mydisc(D,Y^1,Y^2) = \sum_{\mathcal{D},\mathcal{Y}^1,\mathcal{Y}^2}&\sgn(\mathcal{Y}^1)\sgn(\mathcal{Y}^2) \mydisc(D_t,Y^1_t,Y^2_t) \\
    &\prod_{j=1}^k (-1)^{\card{D_{tc_j}}}  \mydisc(D_{tc_j},Y^1_{tc_j},Y^2_{tc_j})\mydisc(D_{cc_j},Y^1_{cc_j},Y^2_{cc_j})
    \end{split}
  \end{align}
  where $\mydisc(\cdot,\cdot,\cdot)$ is as in \eqref{eq:defnpartialdisc}, $\sgn(\cdot,\cdot)$ as in Definition~\ref{defn:parityfunction}, and the sum is over all $\mathcal{D}=(D_t,D_{tc_1},\ldots)$, $\mathcal{Y}^1=(Y^1_t,Y^1_{tc_1},\ldots)$, $\mathcal{Y}^2=(Y^2_t,Y^2_{tc_1},\ldots)$ satisfying:
  \begin{align*}
    Z &=  Z_t \sqcup (Z_{tc_1}\sqcup Z_{cc_1} \sqcup \cdots \sqcup Z_{tc_k}\sqcup Z_{cc_k}) &\mbox{where }Z\in\{D,Y^1,Y^2\}\\
    \zeta(T_{c_j})\setminus  \zeta(t) &\subset Z_{cc_j} \subset \zeta(T_{c_j}) \;\;\;\;
   Z_{t} \subset \zeta(t)  \;\;\;\;
   Z_{tc_j} \subset \zeta(t)\cap\zeta(t_{c_j}) &\mbox{where }\zeta\in\{\alpha,\chi^1,\chi^2\}.
  \end{align*}
\end{lem}

\begin{proof}
  Let $\pi^1:D\to Y^1$ and $\pi^2:D\to Y^2$ be matchings.
  Observe that Lemma~\ref{thm:recursionsign} says that  $\sgn(\pi^1)$ will factor in the tree decomposition, leading to the term $\sgn(\mathcal{D})\sgn(\mathcal{Y}^1)$.
  Similarly, $\sgn(\pi^2)$ leads to the term $\sgn(\mathcal{D})\sgn(\mathcal{Y}^2)$ (note that $\sgn(\mathcal{D})$ cancels).
  Therefore, for the rest of the proof we can ignore all sign factors.
  We can think of the pair $(\pi^1,\pi^2)$ as a subset of $D\times Y^1\times Y^2$.
  In a similar way as we did in the proof of Lemma~\ref{thm:recursepermbipart2}, there is a unique decomposition of $(\pi^1,\pi^2)$ of the form
  \begin{align}\label{eq:partitionfullrhopi}
    (\pi^1,\pi^2) = (\pi^1_t,\pi^2_t) \sqcup (\pi^1_{c_1},\pi^2_{c_1})\sqcup \cdots\sqcup (\pi^1_{c_k},\pi^2_{c_k})
  \end{align}
  where $(\pi^1_t,\pi^2_t)\subset (\alpha\times\chi^1\times \chi^2)(t)$, and for each child $c$ we have that $(\pi^1_c,\pi^2_{c}) \subset (\alpha\times\chi^1\times \chi^2)(T_c)$ and
  \begin{align}\label{eq:overcountrhopi}
    (\pi^1_c,\pi^2_c)\cap (\alpha\times\chi^1\times \chi^2)(t) = \emptyset.
  \end{align}

  Note that by construction $\pi^1_t,\pi^2_t$ have the same domain.
  Let $D_t$ denote this domain, and let $Y^1_t,Y^2_t$ denote the respective ranges.
  Similarly, let $D_c,Y^1_c,Y^2_c$ be the domain and ranges of $\pi^1_c,\pi^2_c$.
  Then we have
  \begin{subequations}\label{eq:partitiondomainrangerhopi}
  \begin{align}
    Z &=  Z_t \sqcup (Z_{c_1} \sqcup \cdots \sqcup Z_{c_k}) &\mbox{where }Z\in\{D,Y^1,Y^2\}\\
   \zeta(T_c)\setminus  \zeta(t) &\subset Z_{c} \subset \zeta(T_c) \;\;\;\;\;\;
   Z_t \subset \zeta(t) &\mbox{where }\zeta\in\{\alpha,\chi^1,\chi^2\}
  \end{align}
  \end{subequations}

  Thus, for matchings $(\pi^1,\pi^2)\in D\times Y^1\times Y^2$ there is a unique partition as in \eqref{eq:partitionfullrhopi} and the corresponding domains and ranges satisfy \eqref{eq:partitiondomainrangerhopi}.
  On the other hand, assume that $D_t,D_c,Y^1_t,Y^1_c,Y^2_t,Y^2_c$  satisfy \eqref{eq:partitiondomainrangepi}, and we are given some matchings $\pi^1_t,\pi^1_c,\pi^2_t,\pi^2_c$ with these domains and ranges and such that  $(\pi^1_c,\pi^2_c)$ satisfy \eqref{eq:overcountrhopi}.
  Then equations \eqref{eq:partitiondomainrangerhopi} tell us that we can merge them into matchings $\pi^1,\pi^2$ with domain $D$ and ranges $Y^1,Y^2$.
  Condition  \eqref{eq:overcountrhopi} ensures we are not overcounting.
  Then we have
  \begin{align}\label{eq:recursivediscDY^1R}
    \mydisc(D,Y^1,Y^2) &= \sum_{D_t,D_c,Y^1_t,Y^1_c,Y^2_t,Y^2_{c}} \mydisc(D_t,Y^1_t,Y^2_t) \prod_{j=1}^k \mydisc^*(D_{c_j},Y^1_{c_j},Y^2_{c_j}) 
  \end{align}
  where the sum is over all triples as in~\eqref{eq:partitiondomainrangerhopi}, and where $\mydisc^*(D_c,Y^1_c,Y^2_c)$ is similar to $\mydisc(D_c,Y^1_c,Y^2_c)$, except that it only uses matchings $(\pi^1_c,\pi^2_c)$ satisfying~\eqref{eq:overcountrhopi}.

  Finally, we can obtain $\mydisc^*(D_c,Y^1_c,Y^2_c)$ using inclusion-exclusion in a similar way as in the proof of Lemma~\ref{thm:recursepermbipart}:
  \begin{align} \label{eq:count_inc_exc_DY^1R}
    \mydisc^*(D_c,Y^1_c,Y^2_c)  
    & = \sum_{
      \substack{
        D_{tc}\sqcup D_{cc} = D_c,\;\;  D_{tc}\subset \alpha(t)\\
        Y^1_{tc}\sqcup Y^1_{cc} = Y^1_c,\;\;  Y^1_{tc}\subset \chi^1(t)\\ 
        Y^2_{tc}\sqcup Y^2_{cc} = Y^2_c,\;\;  Y^2_{tc}\subset \chi^2(t)
      }
    }(-1)^{\card{D_{tc}}}  \mydisc(D_{tc},Y^1_{tc},Y^2_{tc})\mydisc(D_{cc},Y^1_{cc},Y^2_{cc}).
  \end{align}
  Combining equations~\eqref{eq:recursivediscDY^1R} and~\eqref{eq:count_inc_exc_DY^1R}, we obtain equation~\eqref{eq:recursediscDYRfull}.
\end{proof}

We proceed to the complexity analysis.

\begin{thm}\label{thm:discriminant}
  Let $M$ be a list of matrices with associated tripartite graph $G$. 
  Let $(T,\alpha,\chi^1,\chi^2)$ be a tripartite decomposition of $G$ of width $\omega$.
  Then we can compute $\mdisc(M)$ in $\widetilde{O}(n^2 + n\, 3^{\omega})$.
\end{thm}
\begin{proof}
  The proof is very similar to the one of Theorem~\ref{thm:bipartitepermanent}.
  For each node $t\in T$, we compute $\mydisc(D,Y^1,Y^2)$ for every triple $D,Y^1,Y^2$ that satisfies equation~\eqref{eq:validtripleDYR}.
  We will show that for each $t$ we can get $\mydisc(D,Y^1,Y^2)$ for all such $D,Y^1,Y^2$ in $\widetilde{O}((k_t+1)(n+3^{\omega}))$, where $k_t$ is the number of children of $t$. 

  The base case is when $t$ is a leaf of $T$.
  Let $M_0$ be the subarray of $M$ given by indices $(\alpha(t),\chi^1(t),\chi^2(t))$.
  Then we need to find the mixed discriminant of all subarrays of $M_0$.
  We can do this in $\widetilde{O}(2^{\omega})$ using of Lemma~\ref{thm:subdisc}.

  Assume now that $t$ is an internal node of $T$ with $k_t$ children and let $D,Y^1,Y^2$ that satisfy~\eqref{eq:validtripleDYR}.
  Then equation~\eqref{eq:recursediscDYRfull} tells us how to find $\mydisc(D,Y^1,Y^2)$.
  Similarly as in Lemma~\ref{thm:evaluaterecursiondet}, we can evaluate this formula in $\widetilde{O}(k_t(n+3^{\omega}))$, assuming we know all terms in the recursion.
  We already found $\mydisc(D_{cc},Y^1_{cc},Y^2_{cc})$ for all children in the recursion, and  we can find $\mydisc(D_t,Y^1_t,Y^2_t)$ for all $D_t,Y^1_t,Y^2_t$ in $\widetilde{O}(2^{\omega})$ in the same way as for the base case.
  This leads to a bound of $\widetilde{O}((k_t+1)(n+3^{\omega}))$ to compute $\mydisc(D,Y^1,Y^2)$ for all $D,Y^1,Y^2$.
\end{proof}

%%% Local Variables: 
%%% mode: latex
%%% TeX-master: "main"
%%% End: 

% Corrections Pablo -- June 1
% Subset convolution more explicit now
% Complexity determinant 3^n

\subsection{Higher dimensions}

It is easy to see that our methods extend to compute generalizations of the permanent and determinant in higher dimensions.
We consider a square $(d+1)$-dimensional array (or tensor) $M$ of length $n$, i.e., of size $n\times \cdots\times n$ ($d+1$ times).
Here we assume $d$ to be constant.
Let's index the first coordinate of $M$ with a set $A$, and the following coordinates with sets $X^1,\ldots,X^d$.
Consider a function $F$ of the form
\begin{align}\label{eq:defnhigherdimensionperm}
  F(M) = \sum_{\pi_1,\ldots,\pi_{d}}\, \prod_{a\in A} \epsilon_1(\pi_1)\cdots\epsilon_{d}(\pi_{d}) M_{a, \pi_1(a),\ldots,\pi_{d}(a)}
\end{align}
where the sum is over all bijections $\pi_l: A\to X^l$, and where $\epsilon_l(\pi_l)$ is either $1$ or $\sgn(\pi_l)$.

Let's consider two special cases of equation~\eqref{eq:defnhigherdimensionperm}.
The simplest case is when $\epsilon_l(\pi_l)=1$ for all $l$.
We refer to such function as the $(d+1)$-dimensional permanent and we denote it as $\perm(M)$ \cite{Dow1987}.
Some applications of this permanent are shown in \cite{Avgustinovich2010,Tichy2015}.

Consider now the case when $d+1$ is even and $\epsilon_l(\pi_l)=\sgn(\pi_l)$ for all $l$.
This is perhaps the simplest generalization of the determinant, and it is usually referred to as the first Cayley hyperdeterminant \cite{Cayley1846}.
Some applications of the hyperdeterminant are shown in \cite{Luque2003,Barvinok1995}.
As opposed to the $2$-dimensional case, computing the hyperdeterminant is $\sharpp$-hard \cite{Hillar2013}.

We now proceed to extend our decomposition methods to this setting.
We associate a $(d+1)$-partite graph $G$ where for each nonzero entry of $M$ we put a $(d+1)$-clique in the respective coordinates.
A tree decomposition of $G$ can be seen as a tuple $(T,\alpha,\chi^1,\ldots,\chi^d)$.
The width $\omega$ of the decomposition is the largest of $\card{\alpha(t)}+\card{\chi^1(t)}+\cdots+\card{\chi^d(t)}$ among all nodes $t$.

As before, for some sets $D\subset A, Y^l\subset X^l$, we consider the function
\begin{align}\label{eq:defnpartialF}
  f(D,Y^1,\ldots,Y^d) := \sum_{\pi_1,\ldots,\pi_{d}} \,\prod_{a\in A} \epsilon_1(\pi_1)\cdots\epsilon_{d}(\pi_{d}) M_{a, \pi_1(a),\ldots,\pi_{d}(a)}
\end{align}
where the sum is over all bijections $\pi_l: D\to Y^l$.

There are two steps in order to generalize our results to this setting: evaluate $f$ in all subarrays, and evaluate the recursion formula.
For the former, the approach from Lemma~\ref{thm:subperms} (and Lemma~\ref{thm:subdisc}) has a simple generalization.
Indeed, Barvinok shows this for the case of the hyperdeterminant \cite{Barvinok1995}.
The proof is the same for an arbitrary function $F$ as in~\eqref{eq:defnhigherdimensionperm}.
Thus, we have the following.

\begin{prop}[\cite{Barvinok1995}]\label{thm:evaluatehigherperm}
  Let $M_0$ be a $(d+1)$-dimensional array of size $n_0\times\cdots\times n_d$, and let $f$ be as in~\eqref{eq:defnpartialF}. 
  Let $A_0,X_0^1,\ldots,X_0^d$ be its set of coordinates, and let 
  \begin{align*}
    S := \{(D,Y^1,\ldots,Y^d)\subset A_0\times X_0^1\times\cdots\times X_0^d:\, \card{D}=\card{Y^1}=\cdots=\card{Y^d}\}
  \end{align*}
  We can compute $f(D,Y^1,\ldots,Y^d)$ for all tuples in $S$ in $O(n_{max}^{d+1}\,2^{n_0+\cdots+n_d})$, where $n_{max}=\max\{n_0,\ldots,n_d\}$.
\end{prop}

Repeating the same analysis as in the proof of Lemma~\ref{thm:recursedisc}, the recursion formula is:
  \begin{align}\label{eq:recursetensor}
    \begin{split}
      f(D,Y^1,\ldots,Y^{d}) = &\sum_{\mathcal{D},\mathcal{Y}^1,\ldots,\mathcal{Y}^{d}}
      \delta_1(\mathcal{D},\mathcal{Y}^1)\cdots\delta_d(\mathcal{D},\mathcal{Y}^{d}) f(D_t,Y_t^1\ldots,Y_t^{d}) \\
      &\qquad\prod_{j=1}^k (-1)^{\card{D_{tc_j}}}  f(D_{tc_j},Y_{tc_j}^1,\ldots,Y_{tc_j}^{d})f(D_{cc_j},Y_{cc_j}^1,\ldots,Y_{cc_j}^{d})
    \end{split}
  \end{align}
where $\delta_l(\mathcal{D},\mathcal{Y}^l)$ is either $1$ or $\sgn(\mathcal{D})\sgn(\mathcal{Y}^l)$ depending on $\epsilon_l$, and the sum is over all tuples $\mathcal{D},\mathcal{Y}^1,\ldots,\mathcal{Y}^d$ such that
\begin{align*}
  Z &=  Z_t \sqcup (Z_{tc_1}\sqcup Z_{cc_1} \sqcup \cdots \sqcup Z_{tc_k}\sqcup Z_{cc_k}) &\mbox{where }Z\in\{D,Y^1,\ldots,Y^d\}\\
  \zeta(T_{c_j})\setminus  \zeta(t) &\subset Z_{cc_j} \subset \zeta(T_{c_j}) \;\;\;
 Z_{t} \subset \zeta(t)  \;\;\;
 Z_{tc_j} \subset \zeta(t)\cap\zeta(t_{c_j}) &\mbox{where }\zeta\in\{\alpha,\chi^1,\ldots,\chi^d\}
\end{align*}
  
The complexity of the decomposition algorithm is as follows.

\begin{thm}\label{thm:higherdimensions}
  Let $M$ be a square $(d+1)$-dimensional array of length $n$, with $(d+1)$-partite graph $G$. 
  Let $F$ be a generalized determinant/permanent as in~\eqref{eq:defnhigherdimensionperm}. 
  Let $(T,\alpha,\chi^1,\ldots,\chi^d)$ be a tree decomposition of $G$ of width $\omega$.
  Then we can compute $F(M)$ in $\widetilde{O}(n^2 + n\, 3^{\omega})$.
\end{thm}
\begin{proof}
  The proof is very similar to past ones.
  For each node $t$, we compute $f(D,Y^1,\ldots,Y^d)$ for all valid tuples.
  We show that for each $t$ we can do this in $\widetilde{O}((k_t+1)(n+3^{\omega}))$, where $k_t$ is the number of children of $t$. 

  The base case, i.e., leaf nodes, reduces to Proposition~\ref{thm:evaluatehigherperm}, leading to a bound of $\widetilde{O}(2^{\omega})$.

  For an internal node $t$, equation~\eqref{eq:recursetensor} tells us how to find $f(D,Y^1,\ldots,Y^d)$.
  Similarly as in Lemma~\ref{thm:evaluaterecursiondet}, we can evaluate this formula in $O(k_t(n+ 3^{\omega}))$.
\end{proof}

For the special case of the permanent, we can give a better bound.

\begin{thm}
  Let $M$ be a square $(d+1)$-dimensional array of length $n$, with $(d+1)$-partite graph $G$. 
  Let $(T,\alpha,\chi^1,\ldots,\chi^d)$ be a tree decomposition of $G$ of width $\omega$.
  Then we can compute $\perm(M)$ in $\widetilde{O}(n\, 2^{\omega})$.
\end{thm}
\begin{proof}
  If there are no sign factors we can follow the procedure of Lemma~\ref{thm:countpartitions} for the recursion, leading to a bound of $\widetilde{O}((k_t+1)\,2^\omega)$ per node.
\end{proof}

%%% Local Variables: 
%%% mode: latex
%%% TeX-master: "main"
%%% End: 

% Corrections Pablo -- July 7

\section{Mixed volume of zonotopes}\label{s:zonotope}

The mixed volume $\mvol$ of $n$ convex bodies $K_1,\ldots,K_n$ in $\R^n$ is the unique real function that satisfies the following properties.
\begin{itemize}
  \item $\mvol$ is multilinear and symmetric in its arguments.
  \item $\mvol(K,\ldots,K)=n!\,\vol(K)$, where $\vol$ denotes the volume.
\end{itemize}
Alternatively, it can be shown that the function $f(\lambda):=\vol(\sum_i \lambda_i K_i)$ for $\lambda_i\geq 0$, is a homogeneous polynomial, and $\mvol$ is the coefficient of $\lambda_1\cdots \lambda_n$.
%Equivalently, $\mvol$ is the coefficient of $\lambda_1\cdots \lambda_n$ in the expansion of $\vol(\sum_i \lambda_i K_i)$.
For more information about mixed volumes, see e.g., \cite{Schneider2013}.
We focus here in the case that all bodies $K_i$ are zonotopes, which are a special class of polytopes.

\subsection{Mixed volumes and permanents}

\begin{defn}\label{defn:zonotope}
  A \emph{zonotope} $z$ is a polytope that is a Minkowski sum of line segments, i.e., it has the form 
  \begin{align*}
    z &= [0,1]z_1 + [0,1]z_2 + \cdots + [0,1]z_m  = \{r_1 z_1 + \cdots + r_m z_m: 0\leq r_i \leq 1\}
  \end{align*}
  where $z_i\in \R^n$ are vectors.
  In case $z_1,\ldots,z_m$ are linearly independent, we say that $z$ is a \emph{parallelotope}.
\end{defn}

The mixed volume of zonotopes has a simple description as follows.

\begin{prop}
  Let $z^i = \sum_{j\in J_i} [0,1]z_j^i$ be a zonotope, for $i=1,\ldots,n$.
  Then
  \begin{align}\label{eq:mvolzonotopes}
    \mvol(z^1,\ldots,z^n) = \sum_{j_1\in J_1,\ldots,j_n\in J_n}|\det(z_{j_1}^1,z_{j_2}^2,\ldots,z_{j_n}^n)|
  \end{align}
\end{prop}
\begin{proof}
  The multilinearity of the mixed volume implies
  \begin{align*}
    \mvol(z^1,\ldots,z^n)  
    %&= \mvol(\sum_{j_1\in J_1} [0,1]z_{j_1}^1, \ldots, \sum_{j_n\in J_n}[0,1]z_{j_n}^n)
    = \sum_{j_1\in J_1,\ldots,j_n\in J_n}\mvol([0,1]z_{j_1}^1, \ldots,[0,1]z_{j_n}^n).
  \end{align*}
  Thus, we just need to argue that 
  \begin{align*}
    \mvol([0,1]z_{j_1}^1, \ldots,[0,1]z_{j_n}^n) = |\det(z_{j_1}^1,\ldots,z_{j_n}^n)|
  \end{align*}
  which follows by noting that $\sum_i [0,1]\lambda_i z^i_{j_i}$ is a parallelepiped with sides $\lambda_i z^i_{j_i}$, and thus its volume is given by the absolute volume of the determinant.
\end{proof}

The mixed volume of $n$~parallelotopes reduces to a permanent when their main axes are aligned, as shown now.
%For a special case of parallelotopes, computing the mixed volume reduces to a permanent computation, as shown now.

\begin{cor}\label{thm:mvolparallelotopes}
  Let $u_1,\ldots,u_n\in \R^n$ and let $M\in \R^{n\times n}_{\geq 0}$ be a nonnegative matrix.
  Let
  $z^i = \sum_j [0,1]M_{i,j} u_j $
be a zonotope.
Then 
\begin{align*}
  \mvol(z^1,\ldots,z^n) = |\det(u_1,\ldots,u_n)|\;\perm(M).
\end{align*}
\end{cor}
\begin{proof}
  We just need to use equation~\eqref{eq:mvolzonotopes}, and cancel out all terms that contain a repeated vector $u_j$, as the determinant is zero.
  The remaining terms have $|\det(u_1,\ldots,u_n)|$ as a factor and we get the desired formula.
\end{proof}

\subsection{Graph representation}

To use a decomposition method for mixed volumes we need to have a graph description of the zonotopes.
We consider now two different graphs that can be associated to a set of zonotopes, and more generally to polytopes.
The first one is a bipartite graph that can be thought of as the analogue of the bipartite graph of a matrix.
The second one has to do with the sparsity in the standard basis representation and it is a more intuitive notion for general polytopes.

\begin{defn}\label{defn:graphedgepolytopes}
  Let $Q$ be a set of $n$ polytopes in $\R^n$.
  Let $U$ denote the set of all vectors (up to scaling) that are parallel to some edge in $Q$.
  We refer to $U$ as the \emph{edge directions} of $Q$.
  The \emph{edge graph} $G(Q)$ is a bipartite graph with vertices $Q\cup U$ and edges $(q,u)$ if $q$ contains an edge parallel to $u$.
\end{defn}

\begin{defn}\label{defn:graphcoordinatespolytope}
  Let $Q$ be a set of $n$ polytopes in $\R^n$.
  Let $X = \{x_1,\ldots,x_n\}$ denote the coordinates.
  The \emph{coordinates graph} $G^X(Q)$ has $X$ as vertex set, and for each polytope $q\in Q$ we form a clique in all its non constant components.
  Note that if $q=\sum_j [0,1]z_j$ is a zonotope, we form a clique in the nonzero coordinates of $\sum_j z_j$.
\end{defn}

\begin{rem}
  Note that the edge graph $G(Q)$ is invariant under affine transformations of $Q$, whereas the coordinates graph $G^X(Q)$ is not.
\end{rem}

\begin{exmp}[Zonotopes of bounded treewidth]\label{exmp:zonotopestreewidth}
  Consider the following zonotopes:
  \begin{subequations}\label{eq:hardzonotopes}
  \begin{align}
      z^1 &= [0,1] (a_1 e_n + e_1 )  + [0,1] e_1 \\
      z^i &= [0,1] (a_i e_n + e_i -e_{i-1})  + [0,1] (e_i-e_{i-1}) &\mbox{for }i=2,\ldots,n-1\\
      z^n &=[0,1] (a_n e_n - e_{n-1})
  \end{align}
  \end{subequations}
  where $\{e_i\}_i$ is the canonical basis and $a_i\in \Z$ are some integers.

  Note that the segments of the above zonotopes are all nonparallel (there are $2n+1$ edge directions).
  This means that the edge graph $G$ has $n$ connected components, one for each of the zonotopes, and each of component is either a 2-path or a 1-path.
  Thus, $G$ has treewidth~$1$.
  As for the coordinates graph $G^X$, it is the union of the triangles:
  $X_i := \{x_{i},x_{i+1}, x_{n}\}$.
  It is easy to see that $G^X$ has treewidth~$2$.
\end{exmp}

The following example shows the relationship between these graphs and the matrix graphs we used before.

\begin{exmp}[Relationship with matrix graphs]
  Let $z^i= \sum_j [0,1]M_{i,j} u_j$ be zonotopes as in Corollary~\ref{thm:mvolparallelotopes}.
  The edge graph $G$ of the zonotopes has vertices $Z\cup U$ where $Z = \{z^i\}_i$ and $U = \{u_j\}_j$, and it has an edge $(z^i,u_j)$ whenever $M_{i,j} \neq 0$.
  If we replace $Z$ with the row set and $U$ with the column set, this is precisely the bipartite graph of matrix $M$.

  On the other hand, the coordinates graph $G^X$ depends on the sparsity structure of vectors $U$.
  %For instance, if some vector $u\in U$ has no zero components, then $G^X$ is the graph.
  Assume now that $U = \{e_j\}_j$ is the canonical basis.
  Then $G^X$ has an edge $(x_j,x_k)$ whenever there is some $z^i$ with  $M_{i,j}\neq 0$ and $M_{i,k}\neq 0$.
  This corresponds to the column graph of $M$.
\end{exmp}

Because of the above example it is expected that the tree decomposition methods derived for the permanent should allow to compute mixed volumes of certain families of zonotopes.
Indeed, we show now if \emph{there are few edge directions} and the edge graph $G$ has bounded treewidth then the mixed volume can be computed efficiently.

%A consequence of the above example is that we can apply  decomposition techniques to obtain mixed volumes of zonotopes with not too many nonparallel edges.

\begin{thm}\label{thm:mvolzonotopeseasy}
  Let $Z$ be a set of $n$ zonotopes in $\R^n$.
  Let $U$ be the set of edge directions of $Z$, and assume that $d := \card{U} - n$ is constant.
  Let $G$ denote the edge graph of $Z$.
  Given a tree decomposition of $G$ of width~$\omega$, we can compute $\mvol(Z)$ in $\widetilde{O}(n^{d+3} + n^{d+1}\,2^{\omega})$.
\end{thm}
\begin{proof}
  If $|U|<n$ it follows from \eqref{eq:mvolzonotopes} that $\mvol(Z)=0$. If $|U| = n$, then Corollary~\ref{thm:mvolparallelotopes} tells us that we just need to compute the determinant of the $u_i$'s and the permanent of some matrix $M$.
  We can find the determinant in $O(n^3)$ with linear algebra.
  For the permanent, as the edge graph $G$ corresponds to the bipartite graph of $M$, we can use Theorem~\ref{thm:bipartitepermanent}.
  Thus, we can find this permanent in $\widetilde{O}(n\,2^{\omega})$.

  If $|U|>n$, let $W\subset U$ of size $n$.
  For each $z\in Z$ assume $z = \sum_{u\in U} [0,1]c_u u$ for some scalars $c_u$.
  Let $z_{W} =  \sum_{u\in W} [0,1]c_u u$, and let $Z_{W} = \{z_W: z\in Z\}$. 
  It follows from equation~\eqref{eq:mvolzonotopes} that
  \begin{align*}
    \mvol(Z) = \sum_{W\subset U,\; \card{W}=n} \mvol(Z_{W}).
  \end{align*}
  For each such $W$ the associated graph is a subgraph of $G$, so we can compute $\mvol(Z_W)$ in $\widetilde{O}(n^3 + n\,2^{\omega})$.
  As there are ${n+d \choose n} =  O(n^d)$ possible $W$, the result follows.
\end{proof}

\begin{exmp}[Zonotopes with few edge directions]\label{exmp:easyzonotopes}
  Consider the following zonotopes:
  \begin{align*} %\label{eq:easyzonotopes}
      z^i &= [0,1] a_i e_i   + [0,1] b_i e &\mbox{for }i=1,\ldots,n
  \end{align*}
  where $\{e_i\}_i$ is the canonical basis, $e := \sum_j e_j=(1,\ldots,1)$ and $a_i,b_i\in \Z$ are some integers.
  Observe that there are $n+1$ edge directions: $e_1,\ldots,e_n,e$.
  Also note the that the edge graph $G$ is a tree, consisting of pairs $(z^i,e_i)$ and $(z^i,e)$.
  %On the other hand, the coordinates graph $G^X$ is the complete graph.
  Therefore, Theorem~\ref{thm:mvolzonotopeseasy} says that can compute the mixed volume of these polytopes in polynomial time.

  As an application, recall that Bernstein's Theorem~\cite{Bernstein1975} gives a correspondence between mixed volumes and the roots of systems of polynomials.
  This allows us to conclude that we can efficiently count the number of solutions of the following system of equations:
  \begin{align*}
    0 &= c_{i,1}+c_{i,2}\,x_i^{a_i} + c_{i,3}\prod_j x_j^{b_i}+ c_{i,4}\,x_i^{a_i}\prod_j x_j^{b_i} &\mbox{for }i=1,\ldots,n.
  \end{align*}
\end{exmp}

\subsection{Hardness result}
Theorem~\ref{thm:mvolzonotopeseasy} shows that it is possible to exploit tree decompositions for mixed volume computations.
However, it restricts the zonotopes to have a small number of edge directions.
This is a strong requirement which is not satisfied in many cases (such as in Example~\ref{exmp:zonotopestreewidth}). 
Unfortunately, we will see that we need this condition.
We remark that the same condition appears in discrete optimization, where it allows to derive (strongly) polynomial time algorithms for certain discrete convex optimization problems~\cite{Onn2004}.
%We remark that this condition also appears in the context of convex discrete optimization~\cite{Onn2004}.

We show now that computing the mixed volume of the zonotopes in Example~\ref{exmp:zonotopestreewidth} is $\sharpp$-hard.
This shows that mixed volumes of zonotopes continue to be hard, even if both $G$ and $G^X$ have bounded treewidth.
%This shows that neither abstraction $G$ nor $G^X$ allows to efficiently compute mixed volumes based on tree decompositions.
We use a similar reduction as in \cite{Dyer1998}, where they prove that the volume of zonotopes is $\sharpp$-hard.

\begin{lem}\label{thm:determinantsumentries}
  The determinant of the following $n\times n$ matrix is $s_1+s_2+\cdots +s_n$.
  \begin{align*}
    M = \left[\begin{smallmatrix}\\
    1     &-1    &0     &0     &\cdots&0      &0     \\ 
    0     &1     &-1    &0     &\cdots&0      &0     \\ 
    0     &0     &1     &-1    &\cdots&0      &0     \\ 
    \vdots&\vdots&\vdots&\vdots&\ddots&\vdots &\vdots\\\\ 
%    0     &0     &0     &0     &\cdots&0      &0     \\ 
    0     &0     &0     &0     &\cdots&-1     &0     \\ 
    0     &0     &0     &0     &\cdots&1      &-1    \\ 
    s_1   &s_2   &s_3   &s_4   &\cdots&s_{n-1}&s_n 
\end{smallmatrix}\right]
  \end{align*}
\end{lem}
\begin{proof}
  If we perform Gaussian elimination we end up with an upper triangular matrix where the diagonal is: $M_{i,i}=1$ for $i<n$ and $M_{n,n} = s_1+\cdots+s_n$.
\end{proof}

\begin{prop}\label{thm:zonotopeshard}
  The following problem is $\sharpp$-hard.
  Given integers $a_1,\ldots,a_n$, compute the mixed volume of the $n$ zonotopes of equations~\eqref{eq:hardzonotopes}.
\end{prop}
\begin{proof}
  We consider the $\sharpp$-complete problem Subset-Sum:
  given a set of integers $A$, determine the number of subsets $S\subset A$ with sum zero.
  Let $k=n-1$, $A = \{a_1,\ldots,a_k\}$ and let $a_n = \delta$ be a parameter. 
  We will show that the solution to the Subset-Sum problem is given by
  \begin{align}\label{eq:subsetsumzonotopes}
    \frac{1}{2} \mvol(z^1,\ldots,z^n_{\delta=-1}) - \mvol(z^1,\ldots,z^n_{\delta=0}) + \frac{1}{2}\mvol(z^1,\ldots,z^n_{\delta=1}).
  \end{align}

  Let's evaluate equation~\eqref{eq:mvolzonotopes}.
  Consider the following $n\times (2n-1)$ matrix.
  \setcounter{MaxMatrixCols}{20}
  \begin{align*}
    M_\delta = \left[\begin{smallmatrix}\\
    1     &1     &-1    &-1    &0     &0     &\cdots&0     &0     &0     &0     &0     \\ 
    0     &0     &1     &1     &-1    &-1    &\cdots&0     &0     &0     &0     &0     \\ 
    0     &0     &0     &0     &1     &1     &\cdots&0     &0     &0     &0     &0     \\ 
    \vdots&\vdots&\vdots&\vdots&\vdots&\vdots&\ddots&\vdots&\vdots&\vdots&\vdots&\vdots\\\\ 
    0     &0     &0     &0     &0     &0     &\cdots&-1    &-1    &0     &0     &0     \\ 
    0     &0     &0     &0     &0     &0     &\cdots&1     &1     &-1    &-1    &0     \\ 
    0     &0     &0     &0     &0     &0     &\cdots&0     &0     &1     &1     &-1    \\ 
    a_1   &0     &a_2   &0     &a_3   &0     &\cdots&a_{k-1}&0     &a_{k}&0     &\delta 
\end{smallmatrix}\right]
  \end{align*}
  Observe that columns $2i-1$ and $2i$ of $M_\delta$ correspond to $z^i$, and the last column corresponds to $z^n_\delta$.
  Then formula~\eqref{eq:mvolzonotopes} considers submatrices of $M_\delta$ that use columns $j_1,\ldots,j_n$ where $j_i\in \{2i-1,2i\}$ for $i=1,\ldots,k$ and $j_n=2n-1$.
  Note now that for any subset $S\subset A$, there is a natural submatrix $M_\delta^S$ to consider: if $a_i \in S$ then $j_i = 2i-1$ and otherwise $j_i=2i$.
  This correspondence is a bijection.
  Observe also that each submatrix $M_\delta^S$ has the form of Lemma~\ref{thm:determinantsumentries}.
  Thus, we have the following equation:
  \begin{align*}
    \mvol(z^1,\ldots,z^n_\delta) = \sum_{S\subset A} \left|\delta+ \sum_{a_i\in S} a_i\right|.
  \end{align*}
  Finally, observe that for any integer $s$ we have
  \begin{align*}
    \frac{1}{2}\,|(-1)+s|\,-\,|s|\,+\frac{1}{2}\,|1+s| = \begin{cases}
      1 &\mbox{if } s = 0 \\
      0 &\mbox{otherwise}
    \end{cases}
  \end{align*}
  The last two equations imply that \eqref{eq:subsetsumzonotopes} indeed counts the subsets $S$ with sum zero.
\end{proof}

%%% Local Variables: 
%%% mode: latex
%%% TeX-master: "main"
%%% End: 

%\nocite {baz,fuzz,bong}
\bibliography{../../../groebner/refernces}

\begin{thebibliography}{10}

\bibitem{Arnborg1987}
Stefan Arnborg, Derek~G Corneil, and Andrzej Proskurowski.
\newblock Complexity of finding embeddings in a $k$-tree.
\newblock {\em SIAM Journal on Algebraic Discrete Methods}, 8(2):277--284,
  1987.

\bibitem{Avgustinovich2010}
S.~V. Avgustinovich.
\newblock Multidimensional permanents in enumeration problems.
\newblock {\em Journal of Applied and Industrial Mathematics}, 4(1):19--20,
  2010.

\bibitem{Balaji2015}
Nikhil Balaji and Samir Datta.
\newblock Bounded treewidth and space-efficient linear algebra.
\newblock In {\em Theory and Applications of Models of Computation}, volume
  9076 of {\em Lecture Notes in Computer Science}, pages 297--308. Springer
  International Publishing, 2015.

\bibitem{Bapat1989}
R.~B. Bapat.
\newblock Mixed discriminants of positive semidefinite matrices.
\newblock {\em Linear Algebra and its Applications}, 126:107--124, 1989.

\bibitem{Barvinok1995}
Alexander Barvinok.
\newblock New algorithms for linear $k$-matroid intersection and matroid
  $k$-parity problems.
\newblock {\em Mathematical Programming}, 69(1-3):449--470, 1995.

\bibitem{Barvinok1996}
Alexander Barvinok.
\newblock Two algorithmic results for the traveling salesman problem.
\newblock {\em Mathematics of Operations Research}, 21(1):65--84, 1996.

\bibitem{Barvinok1997}
Alexander Barvinok.
\newblock Computing mixed discriminants, mixed volumes, and permanents.
\newblock {\em Discrete \& Computational Geometry}, 18(2):205--237, 1997.

\bibitem{Barvinok2011}
Alexander Barvinok and Alex Samorodnitsky.
\newblock Computing the partition function for perfect matchings in a
  hypergraph.
\newblock {\em Combinatorics, Probability and Computing}, 20(06):815--835,
  2011.

\bibitem{Bernstein1975}
D.~N. Bernstein.
\newblock The number of roots of a system of equations.
\newblock {\em Functional Analysis Appl.}, 9(3):1--4, 1975.

\bibitem{Bjoerklund2007}
Andreas Bj{\"o}rklund, Thore Husfeldt, Petteri Kaski, and Mikko Koivisto.
\newblock Fourier meets {M}{\"o}bius: fast subset convolution.
\newblock In {\em Proceedings of the thirty-ninth annual ACM symposium on
  Theory of computing}, pages 67--74. ACM, 2007.

\bibitem{Blair1993}
Jean~RS Blair and Barry Peyton.
\newblock An introduction to chordal graphs and clique trees.
\newblock In {\em Graph theory and sparse matrix computation}, pages 1--29.
  Springer, 1993.

\bibitem{bodlaender2008combinatorial}
Hans~L Bodlaender and Arie~MCA Koster.
\newblock Combinatorial optimization on graphs of bounded treewidth.
\newblock {\em The Computer Journal}, 51(3):255--269, 2008.

\bibitem{Buergisser2000}
Peter B{\"u}rgisser.
\newblock The computational complexity of immanants.
\newblock {\em SIAM Journal on Computing}, 30(3):1023--1040, 2000.

\bibitem{Cayley1846}
Arthur Cayley.
\newblock M{\'e}moire sur les hyperd{\'e}terminants.
\newblock {\em Journal f{\"u}r die reine und angewandte Mathematik}, 30:1--37,
  1846.

\bibitem{Codenotti1997}
Bruno Codenotti, Valentino Crespi, and Giovanni Resta.
\newblock On the permanent of certain $(0, 1)$ {T}oeplitz matrices.
\newblock {\em Linear Algebra and its Applications}, 267:65--100, 1997.

\bibitem{Courcelle2001}
Bruno Courcelle, Johann~A. Makowsky, and Udi Rotics.
\newblock On the fixed parameter complexity of graph enumeration problems
  definable in monadic second-order logic.
\newblock {\em Discrete Applied Mathematics}, 108(1):23--52, 2001.

\bibitem{Dechter2003}
Rina Dechter.
\newblock {\em Constraint processing}.
\newblock Morgan Kaufmann, 2003.

\bibitem{Dow1987}
Stephen~J Dow and Peter~M Gibson.
\newblock Permanents of $d$-dimensional matrices.
\newblock {\em Linear Algebra and its Applications}, 90:133--145, 1987.

\bibitem{Dyer1998}
Martin Dyer, Peter Gritzmann, and Alexander Hufnagel.
\newblock On the complexity of computing mixed volumes.
\newblock {\em SIAM Journal on Computing}, 27(2):356--400, 1998.

\bibitem{Flarup2007}
Uffe Flarup, Pascal Koiran, and Laurent Lyaudet.
\newblock On the expressive power of planar perfect matching and permanents of
  bounded treewidth matrices.
\newblock In {\em Algorithms and Computation}, volume 4835 of {\em Lecture
  Notes in Computer Science}, pages 124--136. Springer, 2007.

\bibitem{Gurvits2005}
Leonid Gurvits.
\newblock On the complexity of mixed discriminants and related problems.
\newblock In {\em Mathematical Foundations of Computer Science}, pages
  447--458. Springer, 2005.

\bibitem{Hillar2013}
Christopher~J Hillar and Lek-Heng Lim.
\newblock Most tensor problems are {NP}-hard.
\newblock {\em Journal of the ACM (JACM)}, 60(6):1--45, 2013.

\bibitem{Jerrum2004}
Mark Jerrum, Alistair Sinclair, and Eric Vigoda.
\newblock A polynomial-time approximation algorithm for the permanent of a
  matrix with nonnegative entries.
\newblock {\em Journal of the ACM (JACM)}, 51(4):671--697, 2004.

\bibitem{Kasteleyn1967}
Pieter~W Kasteleyn.
\newblock Graph theory and crystal physics.
\newblock {\em Graph theory and theoretical physics}, 1:43--110, 1967.

\bibitem{Luque2003}
Jean-Gabriel Luque and Jean-Yves Thibon.
\newblock {H}ankel hyperdeterminants and {S}elberg integrals.
\newblock {\em Journal of Physics A: Mathematical and General},
  36(19):5267--5292, 2003.

\bibitem{Meer2011}
Klaus Meer.
\newblock An extended tree-width notion for directed graphs related to the
  computation of permanents.
\newblock In {\em Computer Science Theory and Applications}, pages 247--260.
  Springer, 2011.

\bibitem{Minc1987}
Henryk Minc.
\newblock Permanental compounds and permanents of $(0, 1)$-circulants.
\newblock {\em Linear Algebra and its Applications}, 86:11--42, 1987.

\bibitem{Onn2004}
Shmuel Onn and Uriel~G Rothblum.
\newblock Convex combinatorial optimization.
\newblock {\em Discrete \& Computational Geometry}, 32(4):549--566, 2004.

\bibitem{Ryser1963}
Herbert~John Ryser.
\newblock Combinatorial mathematics.
\newblock {\em New York}, 1963.

\bibitem{Schneider2013}
Rolf Schneider.
\newblock {\em Convex bodies: the {B}runn-{M}inkowski theory}, volume~44 of
  {\em Encyclopedia of Mathematics and its Applications}.
\newblock Cambridge University Press, Cambridge, 1993.

\bibitem{Schwartz2009}
Moshe Schwartz.
\newblock Efficiently computing the permanent and {H}afnian of some banded
  {T}oeplitz matrices.
\newblock {\em Linear Algebra and its Applications}, 430(4):1364--1374, 2009.

\bibitem{Temme2012}
Kristan Temme and Pawel Wocjan.
\newblock Efficient computation of the permanent of block factorizable
  matrices.
\newblock {\em arXiv preprint arXiv:1208.6589}, 2012.

\bibitem{Temperley1961}
HNV Temperley and Michael~E Fisher.
\newblock Dimer problem in statistical mechanics---an exact result.
\newblock {\em Philosophical Magazine}, 6(68):1061--1063, 1961.

\bibitem{Tichy2015}
Malte~C Tichy.
\newblock Sampling of partially distinguishable bosons and the relation to the
  multidimensional permanent.
\newblock {\em Physical Review A}, 91(2):022316, 2015.

\bibitem{Valiant1979}
Leslie~G Valiant.
\newblock The complexity of computing the permanent.
\newblock {\em Theoretical computer science}, 8(2):189--201, 1979.

\bibitem{Rooij2009}
Johan~MM van Rooij, Hans~L Bodlaender, and Peter Rossmanith.
\newblock Dynamic programming on tree decompositions using generalised fast
  subset convolution.
\newblock In {\em Algorithms-ESA 2009}, volume 5757 of {\em Lecture Notes in
  Computer Science}, pages 566--577. Springer, 2009.

\bibitem{Vontobel2013}
Pascal~O Vontobel.
\newblock The {B}ethe permanent of a nonnegative matrix.
\newblock {\em Information Theory, IEEE Transactions on}, 59(3):1866--1901,
  2013.

\bibitem{Watanabe2010}
Yusuke Watanabe and Michael Chertkov.
\newblock Belief propagation and loop calculus for the permanent of a
  non-negative matrix.
\newblock {\em Journal of Physics A: Mathematical and Theoretical},
  43(24):242002, 2010.

\end{thebibliography}
\bibliographystyle{plain}

\end{document}